\documentclass[preprints,article,submit,pdftex,moreauthors]{Definitions/mdpi} 
\let\linenumbers\relax
\newcommand{\tr}{\mathrm{tr}}

\usepackage{braket}
\theoremstyle{plain}   
\newtheorem{thm}{Theorem}[section]

\newtheorem{lem}{Lemma}[section]

\newtheorem{defn}{Definition}[section]

\theoremstyle{remark}

\DeclareMathOperator{\Var}{Var}

\firstpage{1} 
\pubvolume{1}
\issuenum{1}
\articlenumber{0}
\pubyear{2026}
\copyrightyear{2026}
\datereceived{ } 
\daterevised{ } 
\dateaccepted{ } 
\datepublished{ } 

\Title{Optimal Nonparametric Estimation of Phase-Space Representations for Non-Gaussian Continuous Variable Quantum States}

\Author{V. Orlov $^{1}$*\orcidA{} and  L. Markovich$^{2,3}$*\orcidB{}}

\AuthorNames{V. Orlov, L. Markovich}

\address{%
$^{1}$ \quad Russian Quantum Center, Skolkovo, Moscow 121205, Russia; vorlovac@outlook.com\\
$^{2}$ \quad Instituut-Lorentz, Universiteit Leiden, P.O. Box 9506, 2300 RA Leiden, The Netherlands\\
$^{3}$$\langle \text{aQa}^\text{L} \rangle$ Applied Quantum Algorithms, Leiden, The Netherlands; markovich@mail.lorentz.leidenuniv.nl}

\corres{Correspondence: vorlovac@outlook.com (V.O.), markovich@mail.lorentz.leidenuniv.nl (L.M.)}

\abstract{We further develop Kernel Quantum State Estimation (KQSE), a fully data-driven nonparametric method for continuous variable quantum state reconstruction and characterization, introduced in our recent work~\cite{markovich2025nonparametric} and rooted in nonparametric kernel density estimation (KDE). Unlike approaches relying on finite-dimensional basis truncations, parametric ans\"atze, or prior models, KQSE combines a new kernel estimator of the characteristic function of the tomographic quadrature distribution with suitable kernel integral transformations. The characteristic function is estimated directly from experimental homodyne or heterodyne data and subsequently transformed to estimate quantum state representations and characteristics, including the Wigner function and the density matrix kernel. Observing that several other physically relevant representations and characteristics admit transformations of a closely related form, we derive convergence rates for the corresponding broad class of KQSE-based estimators. The resulting framework covers all phase-space representations, the photon-number tomogram, trace products of quantum states, and purity. We derive mean squared error convergence rates for these kernel transformed estimators and show that the corresponding KQSE applications inherit the near optimal rate $\widetilde{O}(T^{-1})$, where $T$ is the total number of measurements. The proposed theory applies equally to Gaussian and non-Gaussian continuous variable quantum states without imposing a Fock space cutoff. Numerical experiments on simulated Gaussian and non-Gaussian states and real homodyne data demonstrate the advantages of KQSE over state-of-the-art methods, establishing it as a statistically consistent and computationally modest framework for continuous variable quantum state estimation and characterization, particularly in the non-Gaussian regime.}


\begin{document}


\section{Introduction}


\par Kernel Quantum State Estimation (KQSE) introduced in Ref.~\cite{markovich2025nonparametric} is a fully data-driven nonparametric approach to continuous variable (CV) non-Gaussian quantum state reconstruction and characterization from tomographic data. Unlike the state of art parametric methods~\cite{Chapman:22,PhysRevLett.127.140502,PhysRevResearch.3.033278,PhysRevApplied.18.044041} that are working under different prior assumptions on the class of the states, the KQSE relies only on the  assumptions enter through standard regularity and tail decay conditions of the \textit{symplectic tomogram} $\mathcal{W}(x|\mu,\nu)$ that is the probability distribution function (PDF) of the rotated quadratures $\boldsymbol{X}_{\mu,\nu}=\mu \boldsymbol{q}+\nu \boldsymbol{p}$ at various  phase-space directions  $\mu,\nu\in  \mathbb{R}$~\cite{mancini1996symplectic,ibort2009introduction, manko1997quantum, Dodonov1997, chernega2023dynamics,dudinets2024entangled,manko2020integral}. These conditions are satisfied by the Gaussian-Hermite tomograms that cover most of the experimentally realized optical CV states, leaving out some pathological cases.  
\par The core of the KQSE method is the  nonparametric \textit{kernel characteristic function estimator} (KCFE)~\cite{markovich2025nonparametric}, developed on the basis of the  nonparametric \textit{kernel density estimation} (KDE) method well known in nonparametric statistics~\cite{silverman2018density}. As we show in Ref.~\cite{markovich2024not} the characteristic function (CF) $\phi(t;\mu,\nu)$, being the Fourier transform $e^{itx}$ of the symplectic tomogram, contains the complete information about the quantum state. Moreover, many characteristics of the CV quantum state and its phase-space representations (e.g., Wigner function~\cite{wigner1932quantum}, Husimi function\cite{husimi1940formal}), photon-number tomograms~\cite{Banaszek1996PRL,Banaszek1999PRA,Laiho2010PRL,Nehra2019Optica,Olivares2019NJP}, density matrix elements, trace overlaps, purity, and related state characteristics, are the integral transformation of the CF:
\begin{equation}
\label{eq:general_kernel_representation}
\Theta(z)
=
c_{\Theta}
\iint_{\mathbb R^2}
\phi(1;\mu,\nu)\,
\mathcal{K}_{\Theta}(z;\mu,\nu)\,
d\mu\,d\nu.
\end{equation}
Here $\Theta$ specifies the quantity under consideration, $z$ denotes the corresponding set of variables, $\mathcal{K}_{\Theta}$ is a known transformation kernel, and $c_{\Theta}$ is a normalization constant. For $z=(q,p)\in\mathbb R^2$,  $\mathcal{K}_W(q,p;\mu,\nu)=\exp{(-i(\mu q+\nu p))}$, $c_W=(2\pi)^{-2}$ the transformation defines the Wigner function $W(q,p)$.
\par Consequently, once the tomographic CF is estimated, a broad family of phase-space and quadrature representations can be reconstructed. A natural question arise: how does the accuracy of the KCFE  propagates through such kernel transformations? The answer is not automatic, because each reconstructed object involves an integral transform with its own kernel, normalization, truncation domain, and discretization error.
In Ref.~\cite{markovich2025nonparametric} we show the KQSE rate applied to specific examples of \eqref{eq:general_kernel_representation}, like  the kernel of the density operator, trace of the density operator powers and trace product between two density operators, achieving nearly optimal $\tilde{O}(T^{-1})$ rates in the uniform and $L_2$ norms, where $T$ is the total number of measurement.
\par  In this paper we study the rate of convergence of the KQSE applied to the general integral transform of the CF \eqref{eq:general_kernel_representation}, covering the most important examples like the Wigner and Husimi quasiprobabilities in detail.  We study the total error arising while estimating \eqref{eq:general_kernel_representation} summed from three contributions: truncation of the $(\mu,\nu)$-integration domain, deterministic discretization of the finite integral, and statistical KCFE error on the reconstruction grid.
Under natural light-tail and smoothness assumptions on the tomographic PDF of the state, these three contributions can be controlled explicitly. 
Our main result establishes that the KQSE estimator of the general characteristic-function kernel transformation \eqref{eq:general_kernel_representation} achieves a pointwise mean squared error, uniformly over $z\in\mathcal D_\Theta$, of order $\widetilde O(T_{\mu,\nu}^{-1})$, where $T_{\mu,\nu}$ denotes the total number of tomographic measurements (see Theorem~\ref{thm:total_KQSE_convergence_rate_main}).
 Because this result is established for an entire class of characteristic function transforms rather than for a single target quantity, it provides a unified convergence guarantee for KQSE across a broad range of quantum state characteristics and phase space representations.

\par We illustrate our theory on both synthetic and experimental data. In the synthetic benchmark, homodyne samples are generated from a non-Gaussian coherent cat state, and the same estimated CF is used to reconstruct the Wigner function and photon-number tomograms. The results are also compared with a three component Gaussian mixture maximum-likelihood baseline. We also perform a scaling benchmark with a fixed reconstruction grid and varying sample size per tomographic direction, which isolates the statistical component of the error. Finally, we apply the same reconstruction pipeline to experimental homodyne data for an optical Schr{\"o}dinger kitten state obtained via conditional measurements in a quantum-optical experiment [see Ref.~\cite{lvovsky2002quantumoptical,lvovsky2004iterative}] and compare the reconstructed observables with a reference mixed kitten-state model.
\subsection{Related approaches to quantum state and phase space functions reconstruction}
\par In homodyne and heterodyne experiments, neither the density operator nor phase space functions are measured directly. Instead, one obtains samples from quadrature pdfs, which in homodyne detection correspond to optical tomograms~\cite{vogel1989determination,bertrand1987tomographic,smithey1993measurement,lvovsky2009continuousvariable}. Phase space representations nevertheless provide a complete and often more transparent description of the quantum state, allowing one to visualize its structure, identify nonclassical features such as squeezing, interference fringes, and Wigner negativity~\cite{Albarelli2018Resource}, and evaluate expectation values of observables. Traditionally, the Wigner function is reconstructed by first estimating the quadrature distributions, for example using histograms, and subsequently applying an inverse Radon transform~\cite{leonhardt1995measuring,vogel1989determination}. Such indirect reconstruction is sensitive to sampling noise, while inversion of the underlying integral transform is ill posed and therefore generally requires regularization~\cite{vogel1989determination,lvovsky2009continuousvariable}.  Several  established approaches reconstruct the Wigner function indirectly through the maximum likelihood estimate (MLE) of the underlying quantum state. Iterative maximum likelihood (MaxLik)  estimates a physical density matrix from homodyne data, while the diluted MLE scheme improves the numerical convergence of the likelihood optimization~\cite{banaszek1998maximumlikelihood,lvovsky2004iterative,PhysRevA.75.042108}. Bayesian tomography infers a posterior distribution over density matrices, so the resulting estimate depends on the selected prior~\cite{Chapman:22}. Convex optimization similarly reconstructs a positive, unit trace density matrix in a truncated Fock basis~\cite{PhysRevApplied.18.044041}. These methods provide numerical or posterior optimization procedures, but the cited works do not establish explicit sample dependent convergence rates for the reconstructed density matrix or Wigner function in a specified norm. Machine learning approaches, including conditional generative adversarial networks and deep neural networks, also reconstruct a density matrix or its parametrization and then calculate the corresponding Wigner function~\cite{PhysRevLett.127.140502,PhysRevResearch.3.033278}. However, they also do not provide the convergence guarantees.  
\par In~\cite{naulet2017bayesian}, quantum homodyne tomography is formulated as a Bayesian nonparametric inverse problem. Unlike the finite dimensional reconstruction methods discussed above, the unknown state is not represented by a fixed parameter vector but by a prior on an infinite dimensional function space. The method is therefore structurally different, although it is not assumption free: the selected prior determines the favoured regularity and localization properties and directly affects the posterior convergence. A poorly matched prior may lead to slow convergence or biased reconstruction.
\par In contrast, KQSE was developed with explicit statistical convergence guarantees as a central objective.  KQSE estimates the tomographic CF directly from the quadrature data and propagates its estimation accuracy through the corresponding kernel integral transformations used to reconstruct the  characteristics of the quantum state. The same procedure and convergence theory apply to both Gaussian and non-Gaussian continuous variable states, without changing the model class or imposing a Fock space cutoff. This universality is particularly important for non-Gaussian states, whose oscillatory phase space structure and signatures of nonclassicality can be strongly affected by model misspecification or reconstruction bias.
\par The paper is organized as follows. Section~\ref{subsec:KQSE_general} recalls the tomographic characteristic function and introduces its kernel-transform including the Wigner and Husimi functions, photon-number tomograms, and trace characteristics. Section~\ref{sec:KQSE} summarizes the KCFE construction and its pointwise mean squared error properties. Section~\ref{sec:main_result} develops the main convergence analysis for general kernel quantum-state representations, including truncation, discretization, and statistical error bounds. Section~\ref{sec:numerical_study} presents the synthetic benchmarks, scaling tests, and experimental homodyne data reconstruction. Section~\ref{sec:conclusions} summarizes the results.
\section{Tomographic Characteristic Function and Kernel Quantum-State Representations}
\label{subsec:KQSE_general}
We consider a quantum state in an infinite-dimensional Hilbert space $\mathcal{H}$, represented by the density operator $\boldsymbol{\rho}$. By definition, $\boldsymbol{\rho}$ is Hermitian, positive semidefinite, and has unit trace.  For fixed real parameters $(\mu,\nu)$, the quadrature operator is $ \boldsymbol X_{\mu,\nu}=\mu\boldsymbol q+\nu \boldsymbol p$,
where $\boldsymbol q$ and $\boldsymbol p$ are the canonical position and momentum operators. In continuous variables (CV) quantum systems, the quantum state can be equivalently described by the probability distribution function (PDF) 
\begin{align}\label{0_1}
    \mathcal{W}(x|\mu,\nu)=\mathrm{tr}\bigl[\boldsymbol{\rho}\,\delta(x\,\boldsymbol{1}-\mu\boldsymbol{q}-\nu\boldsymbol{p})\bigr]
\end{align}
 of the measurement outcome $x$ associated with the quadrature observable $\boldsymbol X_{\mu,\nu}$ and  $\delta(\cdot)$ denotes the Dirac delta function of an operator. This PDF is called \textit{symplectic tomograms}~\cite{mancini1996symplectic} and it is always nonnegative and normalized with respect to $x$.   The inverse transformation is
\begin{align}\label{1_1}
\boldsymbol{\rho}=\frac{1}{2\pi}\!\iint
  \left[\int\!\mathcal{W}(x|\mu,\nu)e^{ix}\,dx\right]
e^{-i(\mu\boldsymbol{q}+\nu\boldsymbol{p})}\,d\mu\,d\nu.
\end{align}
The optical tomogram is obtained as a particular parametrization of the symplectic tomogram. 
Indeed, introducing polar coordinates in the parameter plane,
\begin{eqnarray}\label{1359}
\mu=r\cos\theta,\quad 
\nu=r\sin\theta,\quad 
x=ry,\quad 
r\in\mathbb{R}^{+},\quad 
\theta\in[0,2\pi],
\end{eqnarray}
and using the homogeneity property of the symplectic tomogram, one obtains
\begin{eqnarray}
\boldsymbol{\rho}
=
\frac{1}{2\pi}
\int\limits_{-\infty}^{\infty}
\int\limits_{0}^{\infty}
\int\limits_{0}^{2\pi}
r\,
\mathcal{W}(y|\cos\theta,\sin\theta)
\exp\left\{ir\left(y-\boldsymbol{q}\cos\theta-\boldsymbol{p}\sin\theta\right)\right\}
dy\,dr\,d\theta .
\end{eqnarray}
The probability density 
$\mathcal{W}(y|\cos\theta,\sin\theta)$ is called the \textit{optical tomogram}. 
This representation is especially relevant in optical homodyne tomography, where the measured probability distributions correspond precisely to rotated quadratures
$Y_\theta=\boldsymbol{q}\cos\theta+\boldsymbol{p}\sin\theta$.
Thus, the optical tomogram naturally appears as the experimentally accessible form of the symplectic tomogram.
\par In~\cite{markovich2024not} we introduce the characteristic function (CF) of the symplectic tomogram as:
\begin{equation}
\label{eq:phi_def_main}
\phi(t;\mu,\nu)
 \equiv 
\int_{-\infty}^{\infty}
\mathcal{W}(x|\mu,\nu)\,e^{itx}\,dx.
\end{equation}
In the tomographic representation, the value $\phi(t;\mu,\nu)$ plays the central role: it enters the inverse reconstruction formula for the density operator and also determines phase-space representations and trace characteristics~\cite{markovich2024not}. 
However, one need only its value in the fixed point $t=1$ to rewrite \eqref{1_1}:
\begin{equation}
\label{eq:rho_CF_reconstruction_main}
\boldsymbol{\rho}
=
\frac{1}{2\pi}
\iint_{\mathbb R^2}
\phi(1;\mu,\nu)
e^{-i(\mu \boldsymbol q+\nu \boldsymbol p)}
\,d\mu\,d\nu,
\end{equation} 
or to define the Wigner function via \eqref{eq:general_kernel_representation}.
 It is well known that $\forall \mu,\nu$ the CF of any PDF always satisfies:
\begin{eqnarray}
    \phi(0;\mu,\nu)=1,\quad \phi(-t;\mu,\nu)=\phi^{\ast}(t;\mu,\nu), \quad |\phi(t;\mu,\nu)|^2\leq 1.
\end{eqnarray}
However, the tomogram is not an arbitrary PDF. The definition \eqref{0_1} imposes some extra conditions on the CF of the tomogram:
\begin{align}
    \phi(t;-\mu,-\nu)=\phi(-t;\mu,\nu),\quad  \phi(-t;-\mu,-\nu)=\phi(t;\mu,\nu).
\end{align}
Moreover, in Ref.~\cite{markovich2024not}, we formulated conditions on the CF ensuring a physical $\boldsymbol{\rho}$ reconstruction:
\begin{thm}{(Hermiticity, Normalization, Positivity)}\label{thm_3}
The integral \eqref{1_1} defines the density operator corresponding to a quantum state if and only if the CF of the tomogram satisfies
\begin{align}
&\phi(1;\mu,\nu)=\phi(-1;-\mu,-\nu), \quad \forall \mu,\nu, \label{th3_1}\\
&\phi(1;0,0)=1, \label{th3_2}\\
&0\leq \frac{1}{2\pi}
\iint \phi(1;\mu,\nu)\phi_{\chi}^{*}(1;\mu,\nu)\,d\mu d\nu\leq 1,
\quad \forall \ket{\chi}, \label{th3_3}
\end{align}
where $\phi_\chi(t;\mu,\nu)$ is the CF of a pure state $\ket{\chi}$.
\end{thm}

\par In what follows, we shall repeatedly use the symplectic tomogram and the corresponding CF of the $m$-th excited state of the harmonic oscillator 
\begin{eqnarray}\label{excited_HO}
   \!\!\!  \mathcal{W}_m(x|\mu,\nu)
    &=& \mathcal{W}_0(x|\mu,\nu) \frac{1}{2^m m!} H^2_m \left(\frac{x}{\sqrt{\mu^2+\nu^2}}\right),\quad
    {\phi}_m(1;\mu,\nu)
    ={\phi}_0(1;\mu,\nu) L_m\left(\frac{\mu^2+\nu^2}{2}\right),
\end{eqnarray}
where the ground state is a reference Gaussian pair:
\begin{align}\label{1530}
\mathcal{W}_0(x|\mu,\nu)
&=
\frac{1}{\sqrt{\pi(\mu^2+\nu^2)}}
\exp\!\left[-\frac{x^2}{\mu^2+\nu^2}\right],\quad 
\phi_0(t;\mu,\nu)
=
\exp\!\left[-\frac{t^2(\mu^2+\nu^2)}{4}\right].
\end{align}
\par In Refs.~\cite{markovich2024not,orlov2025discrete}, we showed that several physically relevant quantum state quantities, including phase-space representations and trace characteristics, admit integral representations in terms of the tomographic CF. Motivated by their common structure, we introduce in Eq.~\eqref{eq:general_kernel_representation} a unified kernel integral representation for this class of quantities. We next examine this representation through several specific examples, showing how different choices of the integration kernel recover the density matrix kernel, phase space functions, the photon number tomogram, and trace characteristics of quantum states.
\par For two quantum states $\boldsymbol{\rho}_1$ and $\boldsymbol{\rho}_2$ with the
tomographic CFs $\phi_{\boldsymbol{\rho}_1}(1;\mu,\nu)$ and $\phi_{\boldsymbol{\rho}_2}(1;\mu,\nu)$, respectively, the noncommutative
Parseval relation reads
\begin{equation}
\label{eq:trace_overlap_CF_main}
\operatorname{Tr}(\boldsymbol{\rho}_1\boldsymbol{\rho}_2)
=
\frac{1}{2\pi}
\iint_{\mathbb R^2}
\phi_{\boldsymbol{\rho}_1}(1;\mu,\nu)\,
\phi_{\boldsymbol{\rho}_2}^{*}(1;\mu,\nu)\,
d\mu\,d\nu .
\end{equation}
Thus, if $\boldsymbol{\rho}_2$ is fixed as a reference state, the overlap
$\operatorname{Tr}(\boldsymbol{\rho}_1\boldsymbol{\rho}_2)$ is a linear kernel
transform \eqref{eq:general_kernel_representation} of the unknown CF $\phi_{\boldsymbol{\rho}_1}(1;\mu,\nu)$, with $c_{\Theta}=(2\pi)^{-1}$, $\mathcal K_{\Theta}(\mu,\nu)=\phi_{{\boldsymbol{\rho}_2}}^{*}(1;\mu,\nu)$.
As a special case, for $\boldsymbol{\rho}_1=\boldsymbol{\rho}_2=\boldsymbol{\rho}$,
Eq.~\eqref{eq:trace_overlap_CF_main} gives the purity
\begin{equation}
\label{eq:purity_CF_main}
\operatorname{Tr}(\boldsymbol{\rho}^{2})
=
\frac{1}{2\pi}
\iint_{\mathbb R^2}
|\phi(1;\mu,\nu)|^2
\,d\mu\,d\nu .
\end{equation}
Higher trace products, such as
$\operatorname{Tr}(\boldsymbol{\rho}^d)$ and
$\operatorname{Tr}(\prod_{i=1}^{d}\boldsymbol{\rho}_i)$, are multilinear functionals of the CFs and are naturally handled within the analogous CF representational framework~\cite{markovich2024not}.
\par The kernel
transform \eqref{eq:general_kernel_representation} also includes phase-space functions. For $z=(q,p)\in\mathbb R^2$, one has the Wigner function
\begin{equation}
\label{eq:Wigner_kernel_main}
W(q,p)
=
\frac{1}{(2\pi)^2}
\iint_{\mathbb R^2}
\phi(1;\mu,\nu)\,
\mathcal{K}_W(q,p;\mu,\nu)
\,d\mu\,d\nu,
\end{equation}
where the kernel is
\begin{align}\label{W_kernel}
    \mathcal{K}_W(q,p;\mu,\nu)=e^{-i(\mu q+\nu p)}.
\end{align}
Thus, in this case $\Theta=W$, $c_W=(2\pi)^{-2}$.
Another important example, which will be used below, is the \emph{photon-number tomogram}~\cite{10.1063/1.4773139} that is defined by
\begin{eqnarray}\label{dequantizer-ph-number}
\mathsf{w}_r(\alpha)=\tr{[\boldsymbol{\rho} \boldsymbol{D}^\dagger(\alpha) \ket{r}\bra{r}\boldsymbol{D}(\alpha)]},
\end{eqnarray}
dependent on the Weyl displacement operator $
\boldsymbol{D}(\alpha) = \exp\left(\alpha\,\boldsymbol{a}^\dagger - \alpha^{*}\,\boldsymbol{a}\right)$,  $\alpha\in \mathbb{C}$ is a  field amplitude, while $r\in\mathbb N_0$ is the photon-number. Photon-number tomograms are experimentally relevant because they describe the photon-counting statistics of displaced optical states and provide direct access to the Wigner function through displaced photon-number parity, enabling quantum state reconstruction and tests of nonclassicality with photon-number resolving detectors~\cite{Banaszek1996PRL,Banaszek1999PRA,Laiho2010PRL,Nehra2019Optica,Olivares2019NJP}. In the probability representation approach, photon-number tomograms are also used as genuine probability distributions connected to symplectic tomograms and to state characteristics such as fidelity and purity~\cite{Manko2012PNTomography,Manko2019GaussianTomograms}.

In Ref.~\cite{orlov2025discrete}, we showed that the photon-number tomogram can be represented as an integral transform of the tomographic CF:
\begin{equation}
\label{eq:wr_kernel_main}
\mathsf{w}_r(\alpha)
=
\frac{1}{2\pi}
\iint_{\mathbb R^2}
\phi(1;\mu,\nu)\,
\mathcal{K}_r(\alpha;\mu,\nu)
\,d\mu\,d\nu ,
\end{equation}
where
\begin{equation}
\label{eq:Kr_main}
\mathcal{K}_r(\alpha;\mu,\nu)
\equiv
\phi_r(1;\mu,\nu)
\exp\!\left[
-\frac{\nu-i\mu}{\sqrt{2}}\alpha^*
+
\frac{\nu+i\mu}{\sqrt{2}}\alpha
\right].
\end{equation}
Here $\phi_r(1;\mu,\nu)$ is the CF of the $r$-th Fock state given in Eq.~\eqref{excited_HO}. Thus, the photon-number tomogram is of the general form \eqref{eq:general_kernel_representation}, with $\Theta=\mathsf{w}_r$, $z=\alpha$, $c_{\mathsf{w}_r}=(2\pi)^{-1}$,  and $\mathcal{K}_{\Theta}=\mathcal{K}_r$.
\par For completeness, let us also recall the phase-space quantity directly accessed in heterodyne detection. In this measurement scheme, the two orthogonal field quadratures are measured jointly, yielding the complex outcome
$\alpha=(q+ip)/\sqrt{2}$. The corresponding POVM is
$\Pi(\alpha)=\frac{1}{\pi}\ket{\alpha}\bra{\alpha}$,
and the probability density of the outcomes is the Husimi $Q$-function,
\begin{equation}
\label{eq:Husimi_Q_definition}
Q(\alpha)
=
\frac{1}{\pi}\bra{\alpha}\boldsymbol{\rho}\ket{\alpha}.
\end{equation}
The Husimi $Q$-function is in the photon-number tomographic family. Indeed, for $r=0$, the definition adopted here gives $\mathsf{w}_0(\alpha)=\pi Q(-\alpha)$. Thus, the estimation of $\mathsf{w}_0(\alpha)$ also provides an estimate of the Husimi function, up to a known normalization parameter and a reflection of the phase-space argument.
Writing $\alpha=(q+ip)/\sqrt{2}$ with respect to normalization $dq\,dp$, one has $Q(q,p) = \frac{1}{2}
Q\!\left(\alpha\right)=\frac{1}{2} Q\!\left(\frac{q+ip}{\sqrt{2}}\right)$ and 
\begin{equation}
\label{eq:Husimi_CF_kernel}
Q(q,p)
=
\frac{1}{(2\pi)^2}
\iint_{\mathbb R^2}
\phi(1;\mu,\nu)\,
\mathcal{K}_Q((q,p);\mu,\nu)
\,d\mu\,d\nu ,
\end{equation}
where the Husimi kernel is
\begin{equation}
\label{eq:Husimi_kernel_def}
\mathcal{K}_Q(q,p;\mu,\nu)
\equiv 
\phi_0(1;\mu,\nu)\mathcal{K}_W(q,p;\mu,\nu).
\end{equation}
 Thus, the Husimi function also fits the general kernel transform form \eqref{eq:general_kernel_representation}, with $\Theta=Q$, $z=(q,p)$, $c_Q=(2\pi)^{-2}$, and $\mathcal{K}_{\Theta}=\mathcal{K}_Q$.

\par Consequently, all the functions above and  other quasiprobability functions on the phase space~\cite{manko2020integral} fit into the same kernel transform framework \eqref{eq:general_kernel_representation}. This observation allows us to analyze their estimation errors within a single KQSE scheme.

\section{Kernel Quantum State Estimation}\label{sec:KQSE}
The KQSE framework  introduced in Ref.~\cite{markovich2025nonparametric} is as a fully data-driven nonparametric method for reconstructing CV quantum states from tomographic data. It is built on the kernel characteristic function estimator (KCFE), which provides a nonparametric estimate of the CF of the tomographic PDF directly from quadrature measurements. The resulting KCFE serves as the central statistical object from which the  phase-space functions, kernel and trace characteristics of the density operator are estimated through suitable integral transformation~\eqref{eq:general_kernel_representation}. 
\par Suppose that, for each grid point $(\mu_j,\nu_k)$, we are given $n$ independent quadrature samples
$\{X_{1,\mu_j,\nu_k},\dots,X_{n,\mu_j,\nu_k}\}$ drawn from the tomographic PDF $\mathcal{W}(x|\mu_j,\nu_k)$. 
The empirical CF is
\begin{equation}
\label{eq:empirical_CF_main}
\widehat\phi_{X,n}(t;\mu_j,\nu_k)
=
\frac{1}{n}
\sum_{\ell=1}^{n}
e^{itX_{\ell,\mu_j,\nu_k}} .
\end{equation}
The KCFE regularizes this empirical estimator by multiplying it with the CF of a smoothing kernel:
\begin{equation}
\label{eq:KCFE_def_main}
\widehat\phi_{n,h}(t;\mu_j,\nu_k)
=
\widehat\phi_{X,n}(t;\mu_j,\nu_k)\,
\phi_K(th;\mu_j,\nu_k),
\end{equation}
where $h>0$ is the bandwidth parameter and 
\begin{equation}
\label{eq:kernel_CF_def_main}
\phi_K(th;\mu_j,\nu_k)
=
\int_{\mathbb R}
K_{\mu_j,\nu_k}(z)e^{ithz}\,dz .
\end{equation}
Since the kernel $K_{\mu_j,\nu_k}$ is chosen (e.g., any symmetrical PDF) and its choice does not influence the rate of convergence, the factor $\phi_K(th;\mu_j,\nu_k)$ is known analytically. The following result holds~\cite{markovich2025nonparametric}:
\begin{thm}[Mean and variance of the KCFE]~
\label{thm1}
Let $\widehat\phi_{n,h}(t;\mu_j,\nu_k)$ be the KCFE constructed from $n$ independent samples collected at a fixed grid point $(\mu_j,\nu_k)$. Then, for every $t\in\mathbb R$, one has
\begin{equation}
\label{eq:kcfe_mean_main}
\mathbb E[\widehat\phi_{n,h}(t;\mu_j,\nu_k)]
=
\phi(t;\mu_j,\nu_k)\phi_K(th;\mu_j,\nu_k)
\end{equation}
and
\begin{equation}
\label{eq:kcfe_var_main}
\Var\!\bigl(\widehat\phi_{n,h}(t;\mu_j,\nu_k)\bigr)
=
\frac{|\phi_K(th;\mu_j,\nu_k)|^2}{n}
\left(1-\left|\phi(t;\mu_j,\nu_k)\right|^2\right).
\end{equation}
\end{thm}
For complex valued random variables, we use the convention
$\operatorname{Var}(Z)=\mathbb{E}[|Z-\mathbb{E}Z|^{2}]$. Since the CF can take complex values, we define the mean squared error (MSE) using the squared modulus of the complex error as
\begin{eqnarray}\label{eq:M_jk_def_theorem_main}
\operatorname{MSE}\!\left[\widehat{\phi}_{n,h}( t ;\mu_j,\nu_k)\right]&\equiv&\mathbb{E}\left(|\widehat{\phi}_{n,h}( t ;\mu_j,\nu_k)-{\phi}( t ;\mu_j,\nu_k)|^2\right)\\\nonumber
    &=& |\operatorname{Bias}( \widehat{\phi}_{n,h}(t;\mu_j,\nu_k))|^2+ \Var(\widehat{\phi}_{n,h}(t;\mu_j,\nu_k))\\\nonumber
    &=&|\phi_K(th;\mu_j,\nu_k)-1|^2|\phi(t;\mu_j,\nu_k)|^2+\frac{|\phi_K(th;\mu_j,\nu_k)|^2 }{n}\left(1-|\phi(t;\mu_j,\nu_k)|^2\right).
   \end{eqnarray} 
For a suitable bandwidth choice, the KCFE achieves the optimal pointwise MSE rate.
For example, if we choose a Gaussian kernel:
\begin{eqnarray}\label{1004_2}
    K^G_{\mu,\nu}(z) =
     \frac{1}{\sqrt{\pi(\mu^2+\nu^2)}}\exp{\left[-\frac{z^2}{\mu^2+\nu^2}\right]},
\end{eqnarray}
the CF \eqref{eq:kernel_CF_def_main} of this kernel is
\begin{eqnarray}\label{1614}
    \phi^G_K(th;\mu,\nu)=\exp{\left
    [-\frac{t^2h^2(\mu^2+\nu^2)}{4}\right]},
\end{eqnarray}
which is a Gaussian function  in $t$. Therefore the KCFE \eqref{eq:KCFE_def_main} with the Gaussian kernel  is
\begin{eqnarray}\label{1231}
  \widehat{\phi}^G_{n,h}(t;\mu,\nu)=\widehat{\phi}_{X,n}(t;\mu,\nu) \exp{\left[-\frac{t^2h^2(\mu^2+\nu^2)}{4}\right]}.
\end{eqnarray}
The following theorem summarizes the Gaussian KCFE convergence property~\cite{markovich2025nonparametric}:
\begin{thm}[MSE rate of the KCFE with the Gaussian  kernel]
\label{thm2}
Let $\widehat{\phi}^G_{n,h}(t;\mu_j,\nu_k)$ be the KCFE constructed with the Gaussian kernel. 
For fixed $t$, $\mu_j$, and $\nu_k$, and $|\phi(t;\mu_j,\nu_k)|<1$ the bandwidth minimizing the pointwise MSE satisfies
\begin{equation}
\label{eq:kcfe_opt_bandwidth_main}
h_{\phi} = \frac{2}{t\sqrt{\mu^2+\nu^2}}\sqrt{\ln\!\left(1+\frac{1-|\phi(t;\mu,\nu)|^2}{n\,|\phi(t;\mu,\nu)|^2}\right)}=O(n^{-1/2}),
\end{equation}
and the corresponding optimized Gaussian KCFE satisfies
\begin{equation}
\label{eq:kcfe_mse_rate_main}
\mathrm{MSE}\!\left[\widehat{\phi}_{n,h_{\phi}}^G(t;\mu_j,\nu_k)\right] 
   =\frac{(1-|\phi(t;\mu,\nu)|^2)|\phi(t;\mu,\nu)|^2}{n|\phi(t;\mu,\nu)|^2+1-|\phi(t;\mu,\nu)|^2}
= O(1/n).
\end{equation}
\end{thm}
Thus, the KCFE attains an asymptotic MSE rate of the same order, $O(n^{-1})$, as the parametric maximum likelihood estimator (MLE), although the corresponding constants may differ. As in standard KDE, this convergence order is unaffected by the particular choice of a symmetric kernel $\phi_K(th;\mu,\nu)$, provided that the usual regularity conditions are satisfied, the kernel choice may instead affect the constants and finite sample bias and variance. The result is therefore not tied to a specific kernel shape. Moreover, because the KCFE estimates the tomographic CF nonparametrically, without imposing a model on the underlying tomographic PDF, the method applies to both Gaussian and non-Gaussian states, including states with multimodal tomographic distributions. 
A detailed discussion of the latter and the derivation of the optimal bandwidth are given in supplemental materials of the Ref.~\cite{markovich2025nonparametric}.
\section{Main Results}
\label{sec:main_result}
As we mentioned, the quantities estimated in this work can be written in the common integral form \eqref{eq:general_kernel_representation}.
We saw that the different choices of $\mathcal{K}_{\Theta}$ give the phase-space functions, photon-number tomograms, density matrix elements, or trace characteristics. 
This provides a unified way to propagate the KCFE accuracy to a broad class of quantum state representations and functionals. To obtain a computable estimator of \eqref{eq:general_kernel_representation}, we first truncate the integration domain to the rectangular window
\begin{equation}
\label{eq:rectangular_domain_main}
B_{\mu_{\max},\nu_{\max}}
\equiv
[-\mu_{\max},\mu_{\max}]
\times
[-\nu_{\max},\nu_{\max}]
\subset\mathbb R^2 ,
\end{equation}
where $\mu_{\max},\nu_{\max}>0$ denote the largest absolute values of the tomographic parameters accessible in the experiment. These cutoffs are determined by the finite range of measurement settings and by the region in which the corresponding characteristic function can be estimated reliably from the available data. Contributions outside this window are not measured directly and therefore enter the estimator through the truncation error.
The corresponding truncated target function is
\begin{equation}
\label{eq:F_trun_def_main}
\Theta_{\mathrm{trunc}}(z)
\equiv
c_{\Theta}
\iint_{B_{\mu_{\max},\nu_{\max}}}
\phi(1;\mu,\nu)\,
\mathcal{K}_{\Theta}(z;\mu,\nu)
\,d\mu\,d\nu .
\end{equation}
Next, we introduce the uniform grid
$\Delta_{\mu}=2\mu_{\max}/N_{\mu}$ and
$\Delta_{\nu}=2\nu_{\max}/N_{\nu}$,
with
$\mu_j=-\mu_{\max}+j\Delta_{\mu}$ and
$\nu_k=-\nu_{\max}+k\Delta_{\nu}$,
where $j=0,\dots,N_{\mu}-1$ and $k=0,\dots,N_{\nu}-1$. 
The finite-sum approximation of \eqref{eq:F_trun_def_main} is
\begin{equation}
\label{eq:F_tilde_def_main}
\widetilde \Theta_{\mathrm{trunc}}(z)
\equiv
c_{\Theta}\Delta_{\mu}\Delta_{\nu}
\sum_{j=0}^{N_{\mu}-1}
\sum_{k=0}^{N_{\nu}-1}
\phi(1;\mu_j,\nu_k)\,
\mathcal{K}_{\Theta}(z;\mu_j,\nu_k).
\end{equation}
Finally, replacing the unknown grid values $\phi(1;\mu_j,\nu_k)$ by the KCFE \eqref{eq:KCFE_def_main}, gives 
\begin{equation}
\label{eq:F_hat_def_main}
\widehat \Theta_{T_{\mu,\nu}}(z)
:=
c_{\Theta}\Delta_{\mu}\Delta_{\nu}
\sum_{j=0}^{N_{\mu}-1}
\sum_{k=0}^{N_{\nu}-1}
\widehat\phi_{n,h}(1;\mu_j,\nu_k)\,
\mathcal{K}_{\Theta}(z;\mu_j,\nu_k),
\end{equation}
denoting the KQSE estimator of the general kernel transformation $\Theta(z)$ defined in \eqref{eq:general_kernel_representation}, constructed from 
\begin{equation}
\label{eq:total_measurements_main}
T_{\mu,\nu}=nN_{\mu}N_{\nu}.
\end{equation}
number of measurements, where $n$ samples are collected at each tomographic setting $\{\mu_i,\nu_j\}$, $i\in[1,N_{\mu}]$, $j\in[1,N_{\nu}]$.
\par 
The total estimation error therefore separates into three natural contributions:
\begin{equation}
\label{eq:three_term_decomp_main}
\Theta(z)-\widehat \Theta_{T_{\mu,\nu}}(z)
=
\underbrace{\Theta(z)-\Theta_{\mathrm{trunc}}(z)}_{\varepsilon_{\mathrm{trunc}}(z)}
+
\underbrace{\Theta_{\mathrm{trunc}}(z)-\widetilde \Theta_{\mathrm{trunc}}(z)}_{\varepsilon_{\mathrm{dis}}(z)}
+
\underbrace{\widetilde \Theta_{\mathrm{trunc}}(z)-\widehat \Theta_{T_{\mu,\nu}}(z)}_{\varepsilon_{\mathrm{K}}(z)}.
\end{equation}
The first term is the truncation error caused by replacing the full plane by
$B_{\mu_{\max},\nu_{\max}}$. The second is the discretization error caused by replacing the finite integral by a grid sum.  The third is the statistical KCFE error propagated through the kernel map.

This formulation makes the role of each parameter explicit. 
The cutoffs $\mu_{\max}$ and $\nu_{\max}$ control the finite window bias, the grid sizes $N_{\mu}$ and $N_{\nu}$ control the quadrature error, and the sample size $n$ and bandwidth $h$ control the KCFE accuracy. 
Under the smoothness and tail assumptions used below, the three terms in \eqref{eq:three_term_decomp_main} can be balanced so that the resulting KQSE error achieves a nearly optimal rate in the total amount of measurements.

\subsection{Truncation error}
First we study the truncation error $\varepsilon_{\mathrm{trunc}}(z)=\Theta(z)-\Theta_{\mathrm{trunc}}(z)$. Using the definitions of \eqref{eq:F_trun_def_main} and \eqref{eq:F_tilde_def_main}, the truncation error is
given exactly by the contribution of the integral outside the rectangular truncation domain that is upper bounded by
\begin{equation}
\label{eq:trunc_error_abs_bound_main}
|\varepsilon_{\mathrm{trunc}}(z)|
\le
|c_\Theta|
\iint_{\mathbb R^2\setminus B_{\mu_{\max},\nu_{\max}}}
|\phi(1;\mu,\nu)|\,
|\mathcal K_\Theta(z;\mu,\nu)|
\,d\mu\,d\nu .
\end{equation}
Taking the supremum over $z\in\mathcal D_\Theta$, where $\mathcal D_\Theta$ denotes the reconstruction domain of the variable $z$ for the target quantity $\Theta$, we obtain
\begin{equation}
\label{eq:trunc_error_tau_bound_main}
\sup_{z\in\mathcal D_\Theta}
|\varepsilon_{\mathrm{trunc}}(z)|
\le
\tau_\Theta,
\end{equation}
where we used the notation
\begin{equation}
\label{eq:tau_theta_def_main}
\tau_\Theta
\equiv
|c_\Theta|
\iint_{\mathbb R^2\setminus B_{\mu_{\max},\nu_{\max}}}
|\phi(1;\mu,\nu)|\,
\sup_{z\in\mathcal D_\Theta}
|\mathcal K_\Theta(z;\mu,\nu)|
\,d\mu\,d\nu .
\end{equation}
This quantity is finite  if the product of the CF with the observable dependent kernel envelope $\mathcal K_\Theta(z;\mu,\nu)$ that is integrable on the rectangular domain $\mathbb R^2\setminus B_{\mu_{\max},\nu_{\max}}=\{|\mu|>\mu_{\max}\}\cup \{|\nu|>\nu_{\max}\}$.
\par We introduce the kernel weighted integrand
\begin{equation}
\label{eq:F_theta_def_main}
F_\Theta(z;\mu,\nu)
\equiv
\phi(1;\mu,\nu)\mathcal K_\Theta(z;\mu,\nu),
\end{equation}
and to control the truncation error in a form that applies uniformly to all kernel representations considered below, we introduce a kernel weighted envelope function $H_\Theta(\mu,\nu)$. Namely, for a given  $\Theta$, we assume that there exists a nonnegative function $H_\Theta\in L^1(\mathbb R^2)$ such that
\begin{equation}
\label{eq:H_theta_envelope_main}
|F_\Theta(z;\mu,\nu)|
\leq |\phi(1;\mu,\nu)|
\sup_{z\in\mathcal D_\Theta}
|\mathcal K_\Theta(z;\mu,\nu)|
\le
H_\Theta(\mu,\nu).
\end{equation}
Then the truncation error \eqref{eq:tau_theta_def_main} is bounded by the tail of this envelope:
\begin{equation}
\label{eq:tau_theta_H_bound_main}
\tau_\Theta
\le
|c_\Theta|
\iint_{\mathbb R^2\setminus B_{\mu_{\max},\nu_{\max}}}
H_\Theta(\mu,\nu)\,d\mu\,d\nu .
\end{equation}
Since $H_\Theta$ is integrable, this bound vanishes as
$\mu_{\max},\nu_{\max}\to\infty$. The choice of $H_\Theta$ is dependent on $\Theta$. 
\begin{lem}\label{lem_3_1}
Assume that there exist constants $C_\Theta>0$ and $\tau_\mu,\tau_\nu>0$ such that
\begin{equation}
\label{eq:H_theta_exponential_main}
H_\Theta(\mu,\nu)
=
C_\Theta
e^{-\tau_\mu|\mu|-\tau_\nu|\nu|}.
\end{equation}
Then the truncation error \eqref{eq:tau_theta_def_main} is upper bounded by
\begin{equation}
\label{eq:tau_theta_exp_bound_main}
\tau_\Theta
\le
\frac{4|c_\Theta|C_\Theta}{\tau_\mu\tau_\nu}
\left[
e^{-\tau_\mu\mu_{\max}}
+
e^{-\tau_\nu\nu_{\max}}
-
e^{-\tau_\mu\mu_{\max}-\tau_\nu\nu_{\max}}
\right]=
O\!\left(
e^{-\tau_\mu\mu_{\max}}
+
e^{-\tau_\nu\nu_{\max}}
\right).
\end{equation}
\end{lem}
\begin{proof}
By the definition \eqref{eq:tau_theta_def_main} of $\tau_\Theta$ and the envelope bound \eqref{eq:tau_theta_H_bound_main}, we have
\begin{equation}
\label{eq:tau_theta_exp_tail_start_main}
\tau_\Theta
\le
|c_\Theta|C_\Theta
\iint_{\mathbb R^2\setminus B_{\mu_{\max},\nu_{\max}}}
e^{-\tau_\mu|\mu|-\tau_\nu|\nu|}
\,d\mu\,d\nu .
\end{equation}
Since the exponential factor separates, one has
\begin{equation}
\label{eq:full_exp_integral_main}
\iint_{\mathbb R^2}
e^{-\tau_\mu|\mu|-\tau_\nu|\nu|}
\,d\mu\,d\nu
=
\frac{4}{\tau_\mu\tau_\nu},
\end{equation}
and
\begin{equation}
\label{eq:rect_exp_integral_main}
\iint_{B_{\mu_{\max},\nu_{\max}}}
e^{-\tau_\mu|\mu|-\tau_\nu|\nu|}
\,d\mu\,d\nu
=
\frac{4}{\tau_\mu\tau_\nu}
\left(1-e^{-\tau_\mu\mu_{\max}}\right)
\left(1-e^{-\tau_\nu\nu_{\max}}\right).
\end{equation}
Therefore,
\begin{align}
\iint_{\mathbb R^2\setminus B_{\mu_{\max},\nu_{\max}}}
e^{-\tau_\mu|\mu|-\tau_\nu|\nu|}
\,d\mu\,d\nu
&=
\frac{4}{\tau_\mu\tau_\nu}
\left[
1-
\left(1-e^{-\tau_\mu\mu_{\max}}\right)
\left(1-e^{-\tau_\nu\nu_{\max}}\right)
\right].
\end{align}
Substituting this expression into \eqref{eq:tau_theta_exp_tail_start_main} gives
\eqref{eq:tau_theta_exp_bound_main}. 
\end{proof}
\par For the Wigner function, the modulus of the kernel \eqref{W_kernel} is equal to one, and  we may take
$H_W(\mu,\nu)=|\phi(1;\mu,\nu)|$ provided that the state CF is integrable on $\mathbb R^2$.
Then \eqref{eq:H_theta_exponential_main} becomes a direct exponential tail assumption on the CF of the state
\begin{equation}
\label{eq:H_W_main}
|\phi(1;\mu,\nu)|\leq C_{W}e^{-\tau_{\mu}|\mu|-\tau_{\nu}|\nu|}.
\end{equation}
Thus, in the Wigner case, the truncation error is controlled directly by the tail of the state CF. For the Husimi function, the kernel has the form \eqref{eq:Husimi_kernel_def}, so since any CF satisfies $|\phi(1;\mu,\nu)|\le 1$, we can choose
$H_Q(\mu,\nu)
=\phi_0(1;\mu,\nu)$,
which is integrable on $\mathbb R^2$. For the photon-number tomogram, the kernel is given in Eq.~\eqref{eq:wr_kernel_main}, and the displacement dependent exponential factor has modulus one. Therefore, $|\mathcal{K}_r(\alpha;\mu,\nu)|= |\phi_r(1;\mu,\nu)|$.  Hence one can choose an integrable envelope of the form
\begin{equation}
\label{eq:H_r_main}
H_r(\mu,\nu)
=
C_r\left(1+\mu^2+\nu^2\right)^r\phi_0(1;\mu,\nu),
\end{equation}
with a constant $C_r>0$ depending only on $r$.  We can conclude, that the Husimi and photon-number kernels contain Gaussian damping factors. Hence exponential envelope condition \eqref{eq:H_theta_exponential_main} is automatically satisfied by all three functions (with the suitable choice of the constant in \eqref{eq:H_r_main}), because Gaussian decay is faster than exponential decay.
\par For the wide class of the experimentally accessible CV states, the tomographic CF typically exhibits Gaussian type decay in the tomographic parameters. 
In particular, for Hermite--Gaussian PDF arising from harmonic oscillator states and their finite superpositions, the CF is bounded by a Gaussian envelope in $(\mu,\nu)$~\cite{markovich2024not}. 
Such a bound is stronger than the exponential condition \eqref{eq:H_W_main}:
\begin{equation}
\label{eq:gaussian_implies_exp_tail_main}
C e^{-a_\mu\mu^2-a_\nu\nu^2}
\le
C_\Theta e^{-\tau_\mu|\mu|-\tau_\nu|\nu|},
\qquad
(\mu,\nu)\in\mathbb R^2,
\end{equation}
for suitable constants $C_\Theta>0$ and $\tau_\mu,\tau_\nu>0$. 
Therefore, the envelope condition \eqref{eq:H_W_main}
is satisfied for the Wigner, Husimi, and photon-number kernels considered above for any quantum state from the Hermite--Gaussian family, that is basically all the experimentally accessible PDFs in the homodyne/heterodyne experiments.
\par For heavy-tailed photon-count distributions, such as Pareto type laws~\cite{PhysRevLett.123.123606,Mok2026powerlaw}, the exponential envelope do not hold, and the truncation analysis has to be replaced by a more delicate integrable envelope argument adapted to the actual decay of the kernel-weighted CF. 
We do not address such states in the present work, since their treatment requires a separate analysis and, in general, a heavy tailed kernel within the KQSE framework, for example from~\cite{markovich2018light}. 
This extension will be considered in a follow-up paper.

\subsection{Discretization error}
The second contribution in \eqref{eq:three_term_decomp_main} is the
discretization error $\varepsilon_{\mathrm{dis}}(z)=\Theta_{\mathrm{trunc}}(z)-\widetilde \Theta_{\mathrm{trunc}}(z)$.
It is caused by replacing the integral over the finite rectangle
$B_{\mu_{\max},\nu_{\max}}$ by a finite uniform grid sum. 
The error can be written as
\begin{equation}
\label{eq:dis_error_F_main}
\varepsilon_{\mathrm{dis}}(z)
=
c_\Theta
\left[
\iint_{B_{\mu_{\max},\nu_{\max}}}
F_\Theta(z;\mu,\nu)\,d\mu\,d\nu
-
\Delta_\mu\Delta_\nu
\sum_{j=0}^{N_\mu-1}
\sum_{k=0}^{N_\nu-1}
F_\Theta(z;\mu_j,\nu_k)
\right].
\end{equation}
We assume that $F_\Theta(z;\mu,\nu)$ is continuously differentiable in $(\mu,\nu)$ and that its first derivatives are controlled by an integrable envelope. 
Namely, assume that there exists a nonnegative function
$H_{\Theta,1}\in L^1(\mathbb R^2)$ such that
\begin{equation}
\label{eq:F_theta_derivative_envelope_main}
\sup_{z\in\mathcal D_\Theta}
\left(
\left|\partial_\mu F_\Theta(z;\mu,\nu)\right|
+
\left|\partial_\nu F_\Theta(z;\mu,\nu)\right|
\right)
\le
H_{\Theta,1}(\mu,\nu),
\qquad
(\mu,\nu)\in\mathbb R^2 .
\end{equation}
This condition is the differential analogue of the envelope condition \eqref{eq:H_theta_envelope_main} used for the truncation error. 
While $H_\Theta$ controls the tail of the kernel-weighted characteristic function itself, $H_{\Theta,1}$ controls its variation inside the truncated domain and therefore governs the discretization error. It is straightforward to conclude:
\begin{lem}\label{lem_3_2} Under condition \eqref{eq:F_theta_derivative_envelope_main}, the deterministic discretization error satisfies
\begin{eqnarray}
\label{eq:dis_error_bound_envelope_main}
&&
\sup_{z\in\mathcal D_\Theta}
|\varepsilon_{\mathrm{dis}}(z)|
\le
\delta_\Theta,\\
\nonumber
&&
\delta_\Theta
\equiv
\frac{|c_\Theta|}{2}
\left[
(\Delta_\mu+\Delta_\nu)
\iint_{B_{\mu_{\max},\nu_{\max}}}
H_{\Theta,1}(\mu,\nu)\,d\mu\,d\nu
\right]=O\!\left(
\frac{\mu_{\max}}{N_\mu}
+
\frac{\nu_{\max}}{N_\nu}
\right).
\end{eqnarray}
\end{lem}
The derivative envelope condition \eqref{eq:F_theta_derivative_envelope_main} is natural for the class of CV states considered in this work. 
For Gaussian states, Fock states, coherent states, cat states, and any finite superpositions of them, the tomographic CF is given by a Gaussian factor multiplied by polynomials and phase factors~\cite{markovich2024not}. 
Consequently,  $F_\Theta(z;\mu,\nu)$ and its first derivatives admit integrable envelopes in $(\mu,\nu)$. 
For the Wigner kernel this regularity is inherited from the state CF itself, while for the Husimi and photon-number kernels it is further supported by the Gaussian damping present in the reconstruction kernels. 
Thus, the condition excludes heavy-tailed or highly oscillatory cases, but is well suited to the light-tailed CV states treated in the present work.

Stronger discretization rates may be obtained under additional smoothness assumptions.
\begin{defn}
\label{def_3_1}
Let $a_\mu,a_\nu>0$. 
We say that the kernel-weighted CF $F_\Theta(z;\mu,\nu)$
satisfies analytic strip regularity with strip widths $a_\mu$ and $a_\nu$ if, for every
$z\in\mathcal D_\Theta$, it admits an analytic continuation in the variables $\mu$ and $\nu$ to the complex strips
\begin{equation}
\label{eq:analytic_strip_domain_main}
|\operatorname{Im}\mu|<a_\mu,
\qquad
|\operatorname{Im}\nu|<a_\nu,
\end{equation}
and if this continuation is controlled by an integrable strip envelope. 
Namely, we assume that there exists a nonnegative function
$M_{\Theta,a}\in L^1(\mathbb R^2)$ such that
\begin{equation}
\label{eq:analytic_strip_envelope_main}
\sup_{z\in\mathcal D_\Theta}
\left|
F_\Theta\!\left(z;\mu+i s_\mu,\nu+i s_\nu\right)
\right|
\le
M_{\Theta,a}(\mu,\nu)
\end{equation}
for all $(\mu,\nu)\in\mathbb R^2$ and all real $s_\mu,s_\nu$ satisfying
$|s_\mu|<a_\mu$, $|s_\nu|<a_\nu$.
\end{defn}
The analytic strip condition means that $F_\Theta(z;\mu,\nu)$ is not only smooth on the real $(\mu,\nu)$ plane, but remains analytic when the integration variables are shifted slightly into the complex plane. 
The strip widths $a_\mu$ and $a_\nu$ quantify how far such a complex shift can be made before the analytic continuation ceases to be controlled. 
This type of regularity is stronger than the bounded derivative condition used in Lemma~\ref{lem_3_2}. 
Its importance is that analytic functions controlled in a strip have exponentially small high frequency components, which leads to spectral, or exponential, convergence of trapezoidal and Fourier type quadrature rules.
\par For the Wigner, Husimi, and photon-number kernels considered here, the kernels themselves are entire functions of $(\mu,\nu)$. Hence the assumption reduces to the corresponding analytic regularity and strip boundedness of the states CF. 
This is satisfied for all states from the Hermite--Gaussian family, but it is not imposed on the possible in photon-count tomography heavy-tailed states.
\begin{lem}\label{cor_3_1}
Assume that $F_\Theta$ satisfies analytic strip regularity in the sense of Definition~\ref{def_3_1}. 
Assume in addition that, on the finite window 
$B_{\mu_{\max},\nu_{\max}}$, the finite grid approximation is interpreted as a periodic trapezoidal approximation.
Then the discretization error \eqref{eq:dis_error_bound_envelope_main} satisfies
\begin{equation}
\label{eq:analytic_discretization_remark_main}
\delta_\Theta
=
O\!\left(
e^{-\pi a_\mu N_\mu/\mu_{\max}}
+
e^{-\pi a_\nu N_\nu/\nu_{\max}}
\right).
\end{equation}
\end{lem}

\begin{proof}
We split the two-dimensional Riemann sum approximation into two one-dimensional Riemann sum approximations and  introduce 
\begin{equation}
Q_\mu(z)
=
\Delta_\mu
\sum_{j=0}^{N_\mu-1}
\int_{-\nu_{\max}}^{\nu_{\max}}
F_\Theta(z;\mu_j,\nu)\,d\nu .
\end{equation}
Then the discretization error can be written as follows
\begin{eqnarray}\label{1559}
\varepsilon_{\mathrm{dis}}(z)
=
c_\Theta
\left[
\iint_{B_{\mu_{\max},\nu_{\max}}}
F_\Theta(z;\mu,\nu)\,d\mu d\nu
-
Q_\mu(z)
\right]+
c_\Theta
\left[
Q_\mu(z)
-
\Delta_\mu\Delta_\nu
\sum_{j=0}^{N_\mu-1}
\sum_{k=0}^{N_\nu-1}
F_\Theta(z;\mu_j,\nu_k)
\right].
\end{eqnarray}

By the analytic strip assumption, for each fixed $\nu$ the function
$\mu\mapsto F_\Theta(z;\mu,\nu)$ is analytic in the strip
$|\operatorname{Im}\mu|<a_\mu$ and is controlled there by the strip envelope. 
However, on a finite non-periodic interval, analyticity in a strip alone does not by itself imply exponential convergence of a simple rectangular rule, since endpoint mismatch may produce algebraic boundary errors. 
Therefore, we use the additional assumption that the finite window restriction is treated by a periodic trapezoidal rule, or that the endpoint contribution is controlled so that the standard analytic strip trapezoidal estimate applies.

Under this assumption, the one-dimensional trapezoidal estimate for analytic strip functions gives~\cite{trefethen2014exponentially}
\begin{equation}
\left|
\int_{-\mu_{\max}}^{\mu_{\max}}
F_\Theta(z;\mu,\nu)\,d\mu
-
\Delta_\mu
\sum_{j=0}^{N_\mu-1}
F_\Theta(z;\mu_j,\nu)
\right|
\le
C_{\mu,\Theta}
\exp\!\left(
-\frac{2\pi a_\mu}{\Delta_\mu}
\right),
\end{equation}
where $C_{\mu,\Theta}$ is independent of $N_\mu$. 
Since $\Delta_\mu=\tfrac{2\mu_{\max}}{N_\mu}$,
we have
\begin{equation}
\left|
\int_{-\mu_{\max}}^{\mu_{\max}}
F_\Theta(z;\mu,\nu)\,d\mu
-
\Delta_\mu
\sum_{j=0}^{N_\mu-1}
F_\Theta(z;\mu_j,\nu)
\right|
\le
C_{\mu,\Theta}
\exp\!\left(
-\frac{\pi a_\mu N_\mu}{\mu_{\max}}
\right).
\end{equation}
Integrating this estimate over $\nu\in[-\nu_{\max},\nu_{\max}]$ and using the integrable strip envelope, we obtain
\begin{equation}
\left|
\iint_{B_{\mu_{\max},\nu_{\max}}}
F_\Theta(z;\mu,\nu)\,d\mu d\nu
-
Q_\mu(z)
\right|
\le
C'_{\mu,\Theta}
\exp\!\left(
-\frac{\pi a_\mu N_\mu}{\mu_{\max}}
\right),
\end{equation}
where $C'_{\mu,\Theta}$ is independent of $N_\mu$.

Similarly, for each fixed $\mu_j$, the function
$\nu\mapsto F_\Theta(z;\mu_j,\nu)$ is analytic in the strip
$|\operatorname{Im}\nu|<a_\nu$ and satisfies the corresponding strip envelope bound. 
Applying the same analytic strip trapezoidal estimate in the $\nu$ variable gives
\begin{equation}
\left|
Q_\mu(z)
-
\Delta_\mu\Delta_\nu
\sum_{j=0}^{N_\mu-1}
\sum_{k=0}^{N_\nu-1}
F_\Theta(z;\mu_j,\nu_k)
\right|
\le
C'_{\nu,\Theta}
\exp\!\left(
-\frac{\pi a_\nu N_\nu}{\nu_{\max}}
\right),
\end{equation}
where $C'_{\nu,\Theta}$ is independent of $N_\nu$. Combining the two estimates, we find the upper bound of \eqref{1559} to be
\begin{equation}
|\varepsilon_{\mathrm{dis}}(z)|
\le
|c_\Theta|
\left[
C'_{\mu,\Theta}
\exp\!\left(
-\frac{\pi a_\mu N_\mu}{\mu_{\max}}
\right)
+
C'_{\nu,\Theta}
\exp\!\left(
-\frac{\pi a_\nu N_\nu}{\nu_{\max}}
\right)
\right].
\end{equation}
Taking the supremum over $z\in\mathcal D_\Theta$ gives \eqref{eq:analytic_discretization_remark_main}.
\end{proof}

\subsection{KCFE error contribution}
\par     
Next we study the  statistical KCFE error  $\varepsilon_{\mathrm{K}}(z)=\widetilde \Theta_{\mathrm{trunc}}(z)-\widehat \Theta_{T_{\mu,\nu}}(z)$ propagated through the kernel map. Without specifying the kernel $K_{\mu_j,\nu_k}$, we get:
\begin{lem}\label{lem_3_3}
Let $\varepsilon_{\mathrm{K}}(z)$ be the KCFE contribution in the error decomposition
\eqref{eq:three_term_decomp_main}. At each grid point $(\mu_j,\nu_k)$, the KCFE 
$\widehat{\phi}_{n,h}(1;\mu_j,\nu_k)$ is constructed from $n$ independent samples and satisfies the pointwise MSE identity \eqref{eq:M_jk_def_theorem_main}. Then
\begin{eqnarray}
\label{eq:eta_K_bound_main}
&&\sup_{z\in\mathcal D_\Theta}\mathbb E\!\left[
|\varepsilon_{\mathrm{K}}(z)|^2
\right] 
\le \zeta_{\Theta}\\\nonumber
&&\zeta_{\Theta}=
|c_{\Theta}|^2\Delta_{\mu}^2\Delta_{\nu}^2
\sup_{z\in\mathcal D_\Theta}
\left(
\sum_{j=0}^{N_{\mu}-1}
\sum_{k=0}^{N_{\nu}-1}
|\mathcal K_{\Theta}(z;\mu_j,\nu_k)|^2\!
\right)\!\!
\left(
\sum_{j=0}^{N_{\mu}-1}
\sum_{k=0}^{N_{\nu}-1}
  \operatorname{MSE}\!\left[\widehat{\phi}_{n,h}(1;\mu_j,\nu_k)\right]\!\!
\right).
\end{eqnarray}
\end{lem}
\begin{proof}
Using the finite sum formulas \eqref{eq:F_tilde_def_main} and \eqref{eq:F_hat_def_main}, we get
\begin{align}
\varepsilon_{\mathrm K}(z)
=
c_\Theta\Delta_{\mu}\Delta_{\nu}
\sum_{j=0}^{N_{\mu}-1}
\sum_{k=0}^{N_{\nu}-1}
\Bigl(\phi(1;\mu_j,\nu_k)-\widehat\phi_{n,h}(1;\mu_j,\nu_k)\Bigr)
\mathcal{K}_{\Theta}(z;\mu_j,\nu_k).
\end{align}
Applying the Cauchy--Schwarz inequality to the later gives
\begin{align}
 \!\!\!\!\! \mathbb E [|\varepsilon_{\mathrm K}(z)|^2]\leq\!  |c_\Theta|^2\Delta_{\mu}^2\Delta_{\nu}^2\left[\sum_{j=0}^{N_{\mu}-1}
\sum_{k=0}^{N_{\nu}-1}\mathbb E\left[|\phi(1;\mu_j,\nu_k)-\widehat\phi_{n,h}(1;\mu_j,\nu_k)|^2\right]\right] \!\!\left[\sum_{j=0}^{N_{\mu}-1}
\sum_{k=0}^{N_{\nu}-1}|\mathcal{K}_{\Theta}(z;\mu_j,\nu_k)|^2\right].
\end{align}
Using \eqref{eq:M_jk_def_theorem_main}, we get
\begin{equation}
\label{eq:epsilon_K_bound_Mjk_main}
\mathbb E\!\left[
|\varepsilon_{\mathrm{K}}(z)|^2
\right]
\le
|c_{\Theta}|^2\Delta_{\mu}^2\Delta_{\nu}^2
\left(
\sum_{j=0}^{N_{\mu}-1}
\sum_{k=0}^{N_{\nu}-1}
|\mathcal K_{\Theta}(z;\mu_j,\nu_k)|^2
\right)
\left(
\sum_{j=0}^{N_{\mu}-1}
\sum_{k=0}^{N_{\nu}-1}
  \operatorname{MSE}\!\left[\widehat{\phi}_{n,h}( 1 ;\mu_j,\nu_k)\right]
\right).
\end{equation}
Taking the supremum over $z\in\mathcal D_\Theta$, we obtain
\eqref{eq:eta_K_bound_main}. 
\end{proof}
For the Gaussian KCFE with the optimized bandwidth \eqref{eq:kcfe_opt_bandwidth_main}, the pointwise MSE satisfies \eqref{eq:kcfe_mse_rate_main}.
Thus, there exists a constant $C_\phi>0$ such that
\begin{equation}
\label{eq:uniform_kcfe_mse_bound_main}
\operatorname{MSE}\!\left[\widehat{\phi}^{G}_{n,h_\phi}(1;\mu_j,\nu_k)\right]
\le
\frac{C_\phi}{n},
\qquad
j=0,\dots,N_\mu-1,\quad k=0,\dots,N_\nu-1.
\end{equation}
Therefore, we can write
\begin{equation}
\label{eq:sum_kcfe_mse_bound_main}
\sum_{j=0}^{N_\mu-1}
\sum_{k=0}^{N_\nu-1}
\operatorname{MSE}\!\left[\widehat{\phi}^{G}_{n,h_\phi}(1;\mu_j,\nu_k)\right]
\le
\frac{C_\phi N_\mu N_\nu}{n}.
\end{equation}
Substituting \eqref{eq:sum_kcfe_mse_bound_main} into \eqref{eq:eta_K_bound_main} and using $\Delta_\mu=\tfrac{2\mu_{\max}}{N_\mu}$,
$\Delta_\nu=\tfrac{2\nu_{\max}}{N_\nu}$,
we have
\begin{equation}
\label{eq:zeta_gaussian_kcfe_bound_simplified_main}
\zeta_{\Theta}
\le
\frac{
16 C_\phi |c_\Theta|^2
\mu_{\max}^2\nu_{\max}^2
}{nN_\mu N_\nu}
\sup_{z\in\mathcal D_\Theta}
\left(
\sum_{j=0}^{N_\mu-1}
\sum_{k=0}^{N_\nu-1}
|\mathcal K_\Theta(z;\mu_j,\nu_k)|^2
\right)
\end{equation}
for the Gaussian kernel \eqref{1614}.
\begin{lem}\label{cor_3_2}
\label{cor:KCFE_error_bounded_kernel_main}
Assume that the reconstruction kernel is uniformly bounded on the grid, i.e., there exists a constant $\kappa_\Theta>0$ such that
\begin{equation}
\label{eq:kernel_uniform_grid_bound_main}
\sup_{z\in\mathcal D_\Theta}
|\mathcal K_\Theta(z;\mu_j,\nu_k)|
\le
\kappa_\Theta,
\qquad
j=0,\dots,N_\mu-1,\quad k=0,\dots,N_\nu-1 .
\end{equation}
Then the Gaussian KCFE error \eqref{eq:zeta_gaussian_kcfe_bound_simplified_main} is bounded by
\begin{equation}
\label{eq:zeta_gaussian_uniform_kernel_bound_main}
\zeta_{\Theta}
\le
\frac{
16 C_\phi |c_\Theta|^2 \kappa_\Theta^2
\mu_{\max}^2\nu_{\max}^2
}{n}=
O\!\left(
\frac{\mu_{\max}^2\nu_{\max}^2}{n}
\right).
\end{equation}
\end{lem}
For the CF kernel transformation corresponding to the  Wigner  function \eqref{eq:general_kernel_representation}, we have $|\mathcal K_W(q,p;\mu_j,\nu_k)|=1$, and hence the upper bound is simply
\begin{equation}
\label{eq:zeta_wigner_gaussian_kcfe_bound_main}
\zeta_{W}
\le
\frac{
16 C_\phi |c_W|^2
\mu_{\max}^2\nu_{\max}^2
}{n}.
\end{equation}
The same type of bound applies to the Husimi and photon-number cases whenever the corresponding kernel modulus is uniformly bounded on the truncated grid. 
For the Husimi kernel this is immediate because the kernel contains Gaussian damping. 
For the photon-number tomogram, the kernel is a Gaussian factor multiplied by a polynomial, so it is bounded on every finite grid, although the bound may depend on $r$, $\mu_{max}$, and $\nu_{max}$.

\subsection{Total KQSE error convergence rate}
We now combine the three error estimates derived above. 
The following theorem shows that, with a logarithmic growth of the truncation window and a compatible logarithmic refinement of the grid, the total Gaussian KQSE MSE attains a near parametric convergence rate, up to polylogarithmic factors.
\begin{thm}
\label{thm:total_KQSE_convergence_rate_main}
Assume that the conditions of Lemmas~\ref{lem_3_1},~\ref{cor_3_1} and~\ref{cor_3_2} hold.
Let $\widehat{\Theta}_{n,h}(z)$ denote the KQSE estimator of the general
kernel transformation $\Theta(z)$ defined in
\eqref{eq:general_kernel_representation}, constructed from $T_{\mu,\nu}=nN_\mu N_\nu$ quantum state preparations, where $n$ samples are collected per setting
$\{\mu_i,\nu_j\}$, $i\in[1,N_{\mu}]$, $j\in[1,N_{\nu}]$. The accessible reconstruction domain is allowed to expand logarithmically with the number of samples $\mu_{max},\nu_{max}\asymp\log n$, while the reconstruction grid is refined accordingly $N_\mu,N_\nu\asymp(\log n)^2$. 
 Then the estimator achieves
\begin{equation}
\label{eq:total_error_rate_n_tilde_main}
\operatorname{MSE}_{\infty}\!\left(\widehat{\Theta}_{T_{\mu,\nu}}\right)
\equiv
\sup_{z\in\mathcal D_\Theta}
\mathbb E\!\left[
\left|\Theta(z)-\widehat{\Theta}_{T_{\mu,\nu}}(z)\right|^2
\right]=
\widetilde O(T_{\mu,\nu}^{-1}).
\end{equation}
\end{thm}

\begin{proof}
Using the inequality $|a+b+c|^2
\le
3\left(
|a|^2+|b|^2+|c|^2
\right)$, we have the upper bound of the error
$\operatorname{MSE}_{\infty}\!\left(\widehat{\Theta}_{n}\right)
\le
3\left(
\tau_\Theta^2
+
\delta_\Theta^2
+
\zeta_\Theta
\right)$.
By Lemmas~\ref{lem_3_1},~\ref{cor_3_1} and~\ref{cor_3_2}, we obtain
\begin{equation}
\label{1201}
\operatorname{MSE}_{\infty}\!\left(\widehat{\Theta}_{T_{\mu,\nu}}\right)
=
O\!\left(
e^{-2\tau_\mu\mu_{max}}
+
e^{-2\tau_\nu\nu_{max}}
+
e^{-2\pi a_\mu N_\mu/\mu_{max}}
+
e^{-2\pi a_\nu N_\nu/\nu_{max}}
+
\frac{\mu_{max}^2\nu_{\max}^2}{n}
\right).
\end{equation}
As discussed above, the cutoff parameters $\mu_{max}$ and
$\nu_{max}$ determine the finite reconstruction window in the
tomographic parameter space. They must be chosen sufficiently large for
the exponentially decaying tails of the characteristic function to make
the truncation error negligible. At the same time, their growth must be
sufficiently slow to control the statistical KCFE contribution. More precisely, the truncation terms vanish provided that $\mu_{max}\to\infty$, $\nu_{max}\to\infty$, whereas the KCFE contribution vanishes if
$\tfrac{\mu_{max}^2\nu_{max}^2}{n}\to 0$,
or, equivalently,
\begin{equation}
\mu_{max}\nu_{max}
=
o\!\left(n^{1/2}\right).
\end{equation}
A convenient choice satisfying all these requirements is
\begin{equation}
\mu_{max}=A_\mu\log n,
\qquad
\nu_{max}=A_\nu\log n,
\end{equation}
where $A_\mu,A_\nu>0$ are fixed constants. 
It remains to choose the grid sizes $N_\mu$ and $N_\nu$ so that the
analytic strip discretization errors are also polynomially small in
$n$. It is sufficient to take
\begin{equation}
N_\mu
=
\left\lceil B_\mu\mu_{\max}\log n\right\rceil,
\qquad
N_\nu
=
\left\lceil B_\nu\nu_{\max}\log n\right\rceil,
\end{equation}
where $B_\mu,B_\nu>0$ are fixed constants. 
Under this choice of parameters, we can rewrite \eqref{1201} as
\begin{equation}
\label{1202}
\operatorname{MSE}_{\infty}\!\left(\widehat{\Theta}_{T_{\mu,\nu}}\right)
=
O\!\left(
n^{-2\tau_\mu A_\mu}
+
n^{-2\tau_\nu A_\nu}
+
n^{-2\pi a_\mu B_\mu}
+
n^{-2\pi a_\nu B_\nu}
+
\frac{(\log n)^4}{n}
\right).
\end{equation}
We choose the constants to be
\begin{equation}
\label{1204}
2\tau_\mu A_\mu\geq 1,
\qquad
2\tau_\nu A_\nu\geq 1,
\qquad
2\pi a_\mu B_\mu\geq 1,
\qquad
2\pi a_\nu B_\nu\geq 1.
\end{equation}
These conditions are not restrictive from an asymptotic perspective, since they impose only lower bounds on the fixed positive constants $A_\mu,A_\nu,B_\mu$, and $B_\nu$. In particular, they can always be satisfied whenever $\tau_\mu,\tau_\nu,a_\mu,a_\nu>0$. Their role is to ensure that the truncation and discretization errors decay at least as fast as $n^{-1}$. Choosing larger constants does not improve the resulting asymptotic rate, which is determined by the KCFE contribution, but it enlarges the reconstruction window or increases the grid resolution. Therefore, the constants may naturally be selected at, or slightly above, the threshold values \eqref{1204}.  Finally, we get
\begin{equation}
\label{1620}
\operatorname{MSE}_{\infty}\!\left(\widehat{\Theta}_{T_{\mu,\nu}}\right)
=
O\!\left(
\frac{(\log n)^4}{n}
\right)
=
\widetilde O\!\left(n^{-1}\right).
\end{equation}
Under the parameter choice \eqref{1202}, the numbers
of tomographic settings satisfy $N_\mu, N_\nu=\Theta\!\left((\log n)^2\right)$ and the total number of quantum state preparations  is $T_{\mu,\nu}=\Theta\!\left(n(\log n)^4\right)$.
Thus, up to logarithmic factors, $n$ and $T_{\mu,\nu}$ have the same
polynomial scaling, and the convergence rate \eqref{1620} can equivalently
be expressed as \eqref{eq:total_error_rate_n_tilde_main}.
\end{proof}

\section{Numerical Study}
\label{sec:numerical_study}

\subsection{Coherent cat state}
\label{subsec:synthetic_study}

We first illustrate the performance of KQSE on synthetic quadrature
data generated from the three-component coherent cat state (CCS)
\begin{equation}
\label{eq:synthetic_ccs_state}
|\psi_{\mathrm{CCS}}\rangle
=
\mathcal N_c
\sum_{\ell=1}^{3}
|a_\ell\rangle,
\qquad
a_\ell
=
a\exp\!\left(
\frac{2\pi i(\ell-1)}{3}
\right),
\qquad
a=1+0.5i,
\end{equation}
where $|a_\ell\rangle$ are Glauber coherent states and
$\mathcal N_c$ is the normalization constant,
\begin{equation}
\label{eq:ccs_normalization}
|\mathcal N_c|^{-2}
=
\sum_{j,k=1}^{3}
\exp\!\left[
-|a|^2
+
|a|^2
\exp\!\left(
\frac{2\pi i(k-j)}{3}
\right)
\right].
\end{equation}
This state is non-Gaussian and produces multimodal quadrature
distributions, providing a useful benchmark for nonparametric
reconstruction. Using \eqref{1530} and the notation $\beta\equiv\sqrt{\mu^2+\nu^2}$, the symplectic tomogram of the CCS reads
\begin{equation}
\label{eq:ccs_exact_tomogram}
\mathcal W_{\mathrm{CCS}}(x|a,\mu,\nu)
=
\mathcal W_0(x|\beta)
|\mathcal N_c|^2
\sum_{j,k=1}^{3}
\exp\!\left[
ixs_{jk}+d_{jk}
\right],
\end{equation}
where
\begin{equation}
\label{eq:ccs_sjk}
s_{jk}
\equiv
\frac{\sqrt{2}}{\beta^2}
\left[
(\nu-i\mu)a^*
\exp\!\left(
-\frac{2\pi i(k-1)}{3}
\right)
-
(\nu+i\mu)a
\exp\!\left(
\frac{2\pi i(j-1)}{3}
\right)
\right],
\end{equation}
and
\begin{align}
\label{eq:ccs_djk}
d_{jk}
&\equiv
-|a|^2
+
\frac{1}{2\beta^2}
\Bigg[
(\nu+i\mu)^2a^2
\exp\!\left(
\frac{4\pi i(j-1)}{3}
\right) +
(\nu-i\mu)^2(a^*)^2
\exp\!\left(
-\frac{4\pi i(k-1)}{3}
\right)
\Bigg].
\end{align}
The corresponding CF is
\begin{equation}
\label{eq:ccs_exact_CF}
\phi_{\mathrm{CCS}}(t;a,\mu,\nu)
=
|\mathcal N_c|^2
\sum_{j,k=1}^{3}
\phi_0(t+s_{jk};\beta)e^{d_{jk}},
\qquad
\phi_0(t;\beta)
=
\exp\!\left(
-\frac{t^2\beta^2}{4}
\right).
\end{equation}
Using the integral connection of the CF and the Wigner function, we can deduce
\begin{align}
\label{eq:ccs_exact_Wigner}
W_{\mathrm{CCS}}(q,p)
&=
\frac{|\mathcal N_c|^2}{\pi}
\sum_{j,k=1}^{3}
\exp\Bigg[
-2|\gamma|^2
+
2\gamma a_k^*
+
2\gamma^*a_j
-a_ja_k^*
-\frac{|a_j|^2+|a_k|^2}{2}
\Bigg],
\qquad
\gamma=\frac{q+ip}{\sqrt{2}}.
\end{align}
Similarly, the photon-number tomogram corresponding to the convention
\eqref{dequantizer-ph-number} can be written as
\begin{equation}
\label{eq:ccs_exact_photon_tomogram}
\mathsf w^{\mathrm{CCS}}_r(\alpha)
=
\frac{|\mathcal N_c|^2}{r!}
e^{-|a|^2-|\alpha|^2}
\left|
\sum_{\ell=1}^{3}
e^{-\alpha^*a_\ell}
(\alpha+a_\ell)^r
\right|^2,
\qquad
r=0,1,2,\ldots .
\end{equation}
We consider $n\in\{500,1000,2000\}$ quadrature samples per $(\mu,\nu)$ setting. For each value of $n$,
the complete sampling and reconstruction procedure is repeated
independently $R=100$ times. We select $N_\mu=N_\nu=41$, and the corresponding number of
quadrature samples entering one complete reconstruction is
$T_{\mu,\nu}=\{840500,1681000,3362000\}$ for
$n\in\{500,1000,2000\}$, respectively. To quantify the reconstruction accuracy, we use integrated and
uniform pointwise errors. For a single reconstruction
$\widehat{\Theta}^{(s)}$ of a target quantity $\Theta$, the integrated
squared error is
\begin{equation}
\label{eq:ISE_general}
\operatorname{ISE}_{\Theta}^{(s)}
=
\int_{\mathcal D_\Theta}
\left|
\widehat{\Theta}^{(s)}(z)
-
\Theta_{\mathrm{ref}}(z)
\right|^2
\,dz .
\end{equation}
\begin{figure}[h]
\centering
\includegraphics[width=0.85\textwidth]{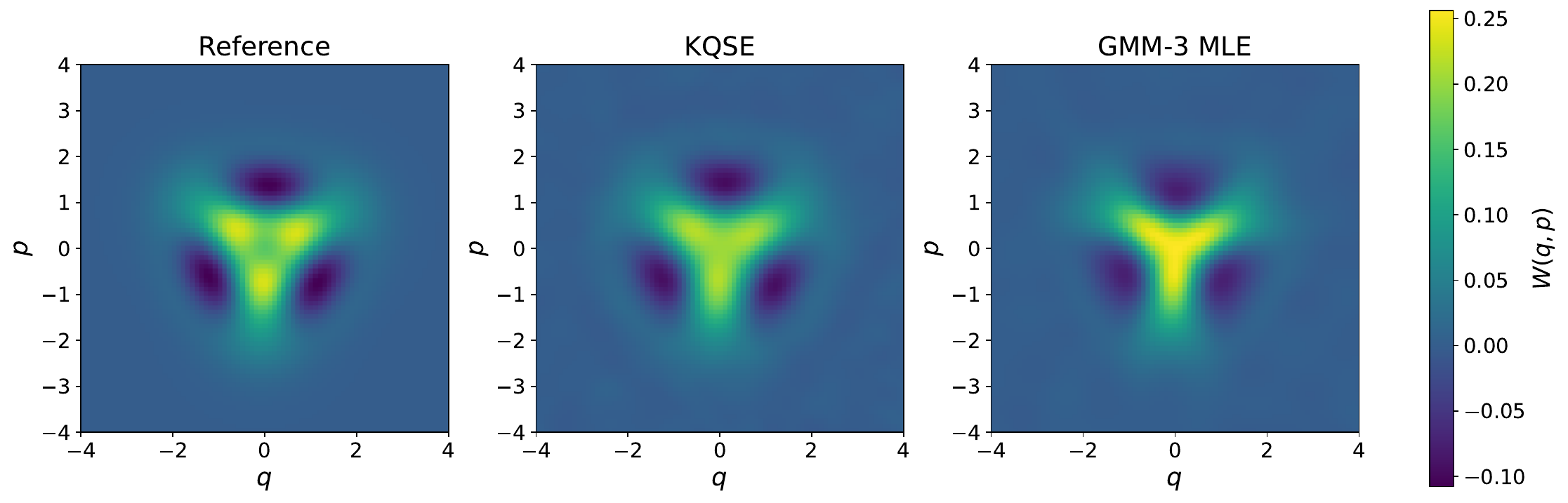}
\caption{
The estimate of the Wigner function  for the CCS
\eqref{eq:synthetic_ccs_state} with $a=1+0.5i$ based on $n=1000$
quadrature samples per tomographic setting. The reference
function \eqref{eq:ccs_exact_Wigner} is compared with KQSE and the GMM-3 MLE
estimators obtained from the same synthetic quadrature data.
}
\label{fig:synthetic_wigner}
\end{figure}
For the synthetic benchmark, averaging over the $R$ independent
realizations gives the empirical mean integrated squared error
\begin{equation}
\label{eq:MISE_general}
\widehat{\operatorname{MISE}}_{\Theta}
=
\frac{1}{R}
\sum_{s=1}^{R}
\widehat{\operatorname{ISE}}_{\Theta}^{(s)} .
\end{equation}
To characterize the largest pointwise statistical error, we also use
the finite-grid empirical uniform pointwise MSE
\begin{equation}
\label{eq:empirical_uniform_pointwise_MSE}
\widehat{\operatorname{MSE}}_{\infty,\Gamma}
\left(\widehat{\Theta}\right)
=
\max_{z_m\in\Gamma_\Theta}
\left\{
\frac{1}{R}
\sum_{s=1}^{R}
\left|
\widehat{\Theta}^{(s)}(z_m)
-
\Theta_{\mathrm{ref}}(z_m)
\right|^2
\right\},
\end{equation}
where $\Gamma_\Theta$ denotes the corresponding finite evaluation
grid. Importantly, the average over independent realizations is taken
at each grid point before the maximum is evaluated, consistently with
the uniform pointwise MSE considered in
Theorem~\ref{thm:total_KQSE_convergence_rate_main}.

Numerically, the integral in Eq.~\eqref{eq:ISE_general} is evaluated
using the same uniform-grid rectangle quadrature as in the
reconstruction: the weight is $\Delta q\,\Delta p$ for the Wigner
function and $\Delta\alpha$ for the photon-number tomograms on the
real-displacement slice.

The KQSE estimated CCS Wigner function and photon-count tomogram were compared with MLE based on parametric Gaussian mixture models (GMM) with three components ($GMM = 3$). Figures~\ref{fig:synthetic_wigner} and
\ref{fig:synthetic_photon} show a fixed reconstruction for
$n=1000$. 
The figures illustrate a single finite sample reconstruction, whereas
all error values reported below are computed from the full ensemble
of $R=100$ independent realizations.
\begin{figure}[h]
\centering
\includegraphics[width=0.8\textwidth]{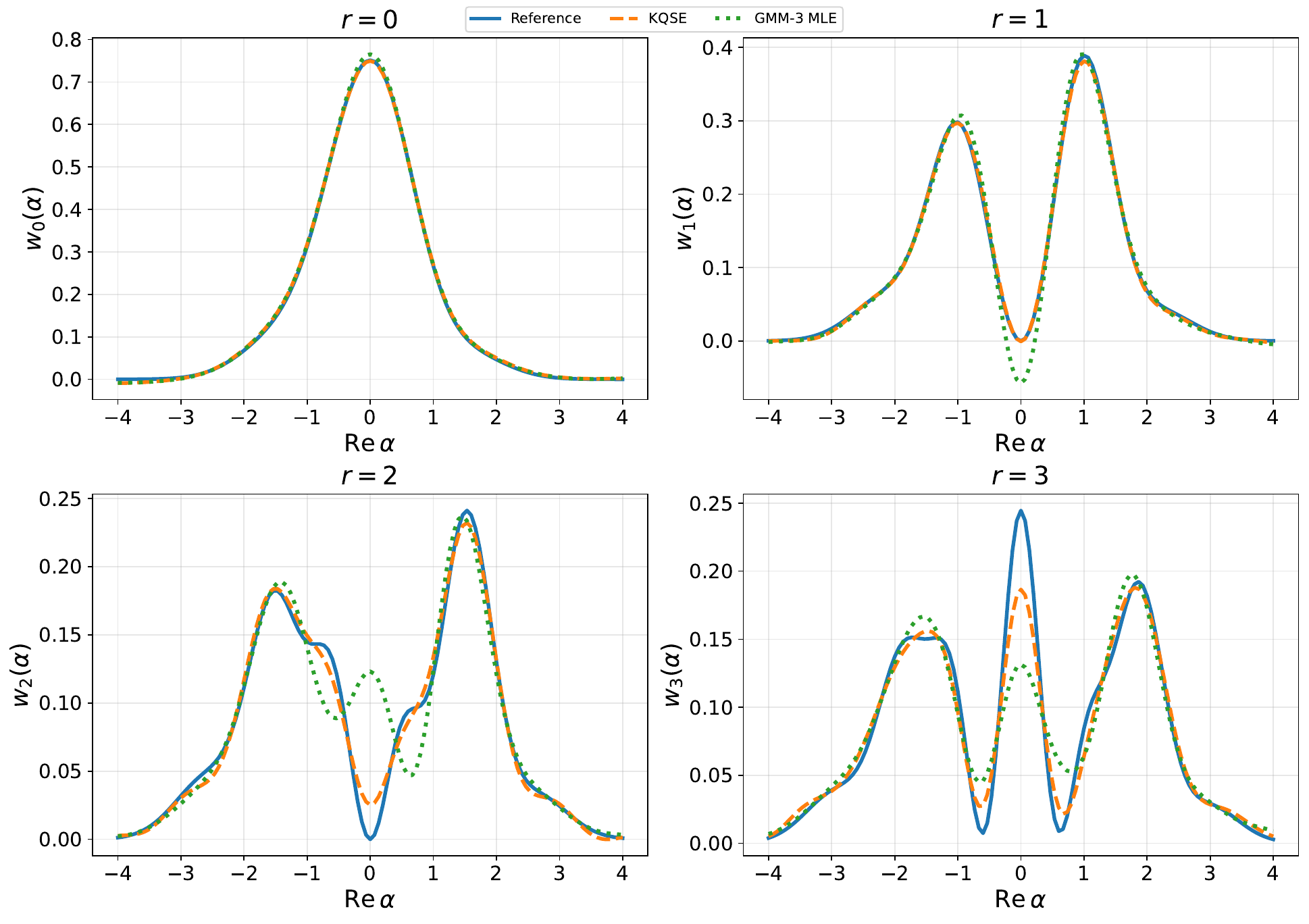}
\caption{
The estimate of the Photon-number tomogram of the CCS  based on $n=1000$ quadrature samples per tomographic setting. The panels show
$\mathsf w_r(\alpha)$, $r=0,1,2,3$, on the real displacement slice
$\operatorname{Im}\alpha=0$. The reference reconstructions \eqref{eq:ccs_exact_photon_tomogram} are compared with KQSE and GMM-3 MLE
estimators obtained from the same synthetic quadrature data.
}
\label{fig:synthetic_photon}
\end{figure}
The MISE \eqref{eq:MISE_general} of the estimators obtained from the $R=100$ independent
realizations are summarized in
Table~\ref{tab:synthetic_mise}.
\begin{table}[H]
\caption{
The MISE of the KQSE and GMM-3 MLE of the Wigner-function and the photon-number
tomogram. Each value is averaged over $R=100$
independent realizations.
}
\label{tab:synthetic_mise}
\centering
\small
\setlength{\tabcolsep}{4pt}
\begin{tabularx}{\textwidth}{@{}CCCCC@{}}
\toprule
\textbf{Function} &
\textbf{Method} &
$\boldsymbol{n=500}$ &
$\boldsymbol{n=1000}$ &
$\boldsymbol{n=2000}$ \\
\midrule
\multirow[m]{2}{*}{$W(q,p)$}
& KQSE      & $4.580\times10^{-3}$ & $2.059\times10^{-3}$ & $7.197\times10^{-4}$ \\
& GMM-3 MLE & $7.088\times10^{-3}$ & $7.275\times10^{-3}$ & $7.452\times10^{-3}$ \\
\midrule
\multirow[m]{2}{*}{$\mathsf w_0(\alpha)$}
& KQSE      & $2.296\times10^{-4}$ & $7.870\times10^{-5}$ & $3.655\times10^{-5}$ \\
& GMM-3 MLE & $2.011\times10^{-4}$ & $1.302\times10^{-4}$ & $1.092\times10^{-4}$ \\
\midrule
\multirow[m]{2}{*}{$\mathsf w_1(\alpha)$}
& KQSE      & $3.688\times10^{-4}$ & $1.204\times10^{-4}$ & $4.779\times10^{-5}$ \\
& GMM-3 MLE & $1.834\times10^{-3}$ & $1.779\times10^{-3}$ & $1.746\times10^{-3}$ \\
\midrule
\multirow[m]{2}{*}{$\mathsf w_2(\alpha)$}
& KQSE      & $1.121\times10^{-3}$ & $5.276\times10^{-4}$ & $1.511\times10^{-4}$ \\
& GMM-3 MLE & $7.057\times10^{-3}$ & $7.496\times10^{-3}$ & $7.710\times10^{-3}$ \\
\midrule
\multirow[m]{2}{*}{$\mathsf w_3(\alpha)$}
& KQSE      & $3.845\times10^{-3}$ & $1.857\times10^{-3}$ & $5.440\times10^{-4}$ \\
& GMM-3 MLE & $5.694\times10^{-3}$ & $6.177\times10^{-3}$ & $6.498\times10^{-3}$ \\
\bottomrule
\end{tabularx}
\end{table}
The corresponding uniform pointwise MSEs
\eqref{eq:empirical_uniform_pointwise_MSE} are given in
Table~\ref{tab:synthetic_mseinf}.

\begin{table}[H]
\caption{
Finite-grid uniform pointwise MSE of the KQSE and GMM-3 MLE of the Wigner-function and the photon-number
tomogram. At every evaluation point, the squared error is
first averaged over the $R=100$ realizations and the maximum is then
taken over the evaluation grid.
}
\label{tab:synthetic_mseinf}
\centering
\small
\setlength{\tabcolsep}{4pt}
\begin{tabularx}{\textwidth}{@{}CCCCC@{}}
\toprule
\textbf{Function} &
\textbf{Method} &
$\boldsymbol{n=500}$ &
$\boldsymbol{n=1000}$ &
$\boldsymbol{n=2000}$ \\
\midrule
\multirow[m]{2}{*}{$W(q,p)$}
& KQSE      & $2.848\times10^{-3}$ & $1.766\times10^{-3}$ & $5.941\times10^{-4}$ \\
& GMM-3 MLE & $8.332\times10^{-3}$ & $8.995\times10^{-3}$ & $9.362\times10^{-3}$ \\
\midrule
\multirow[m]{2}{*}{$\mathsf w_0(\alpha)$}
& KQSE      & $1.240\times10^{-4}$ & $2.705\times10^{-5}$ & $7.983\times10^{-6}$ \\
& GMM-3 MLE & $1.463\times10^{-4}$ & $1.387\times10^{-4}$ & $1.445\times10^{-4}$ \\
\midrule
\multirow[m]{2}{*}{$\mathsf w_1(\alpha)$}
& KQSE      & $2.476\times10^{-4}$ & $5.064\times10^{-5}$ & $1.342\times10^{-5}$ \\
& GMM-3 MLE & $3.284\times10^{-3}$ & $3.303\times10^{-3}$ & $3.305\times10^{-3}$ \\
\midrule
\multirow[m]{2}{*}{$\mathsf w_2(\alpha)$}
& KQSE      & $1.656\times10^{-3}$ & $8.637\times10^{-4}$ & $2.263\times10^{-4}$ \\
& GMM-3 MLE & $1.453\times10^{-2}$ & $1.558\times10^{-2}$ & $1.608\times10^{-2}$ \\
\midrule
\multirow[m]{2}{*}{$\mathsf w_3(\alpha)$}
& KQSE      & $8.420\times10^{-3}$ & $4.045\times10^{-3}$ & $1.163\times10^{-3}$ \\
& GMM-3 MLE & $1.206\times10^{-2}$ & $1.330\times10^{-2}$ & $1.407\times10^{-2}$ \\
\bottomrule
\end{tabularx}
\end{table}

The KQSE errors decrease systematically with increasing sample size
for all reconstructed quantities considered. In contrast, the GMM-3 MLE
reconstruction exhibits clear saturation for the Wigner function and
for $\mathsf w_1$, $\mathsf w_2$, and $\mathsf w_3$. For
$\mathsf w_0$, GMM-3 has a slightly smaller MISE at $n=500$, whereas
KQSE becomes more accurate at $n=1000$ and $n=2000$. This behaviour
is consistent with the expected finite sample variance--bias
trade-off: a restricted parametric model can be competitive at small
sample sizes, but its model approximation bias is not eliminated by
increasing $n$.

\subsection{CCS state not confined to a low-dimensional Fock subspace}
\label{subsec:synthetic_state_methods}
Recent developments in quantum state reconstruction employ a range of techniques, including Bayesian methods for CV systems~\cite{Chapman:22}, neural-network-based machine learning approaches~\cite{PhysRevLett.127.140502,PhysRevResearch.3.033278}, and convex-optimization frameworks~\cite{PhysRevApplied.18.044041}. In most cases, these methods assume a finite-dimensional or otherwise parameterized description of the state, such as a Hilbert space truncation or a neural-network ansatz, and determine the corresponding parameters through potentially high dimensional optimization. Although these strategies can be highly adaptable to experimental data, rigorous statistical guarantees on the convergence of the resulting state estimator are typically unavailable. A more detailed discussion on this is given in the Supplemental Material of Ref.~\cite{markovich2025nonparametric}.
\par In
order to compare KQSE with these state of art  state reconstruction methods we consider a substantially more demanding synthetic benchmark.
We use the same CCS~\eqref{eq:synthetic_ccs_state} but with the bigger coherent amplitude $a=3+0.5i$. The common synthetic homodyne data set contains
$N_{\theta}=72$ phase settings with $n=2000$ quadrature samples per
setting, corresponding to $T=144000$ measurements in total.
\begin{figure}[ht]
\centering
\includegraphics[width=0.8\textwidth]
{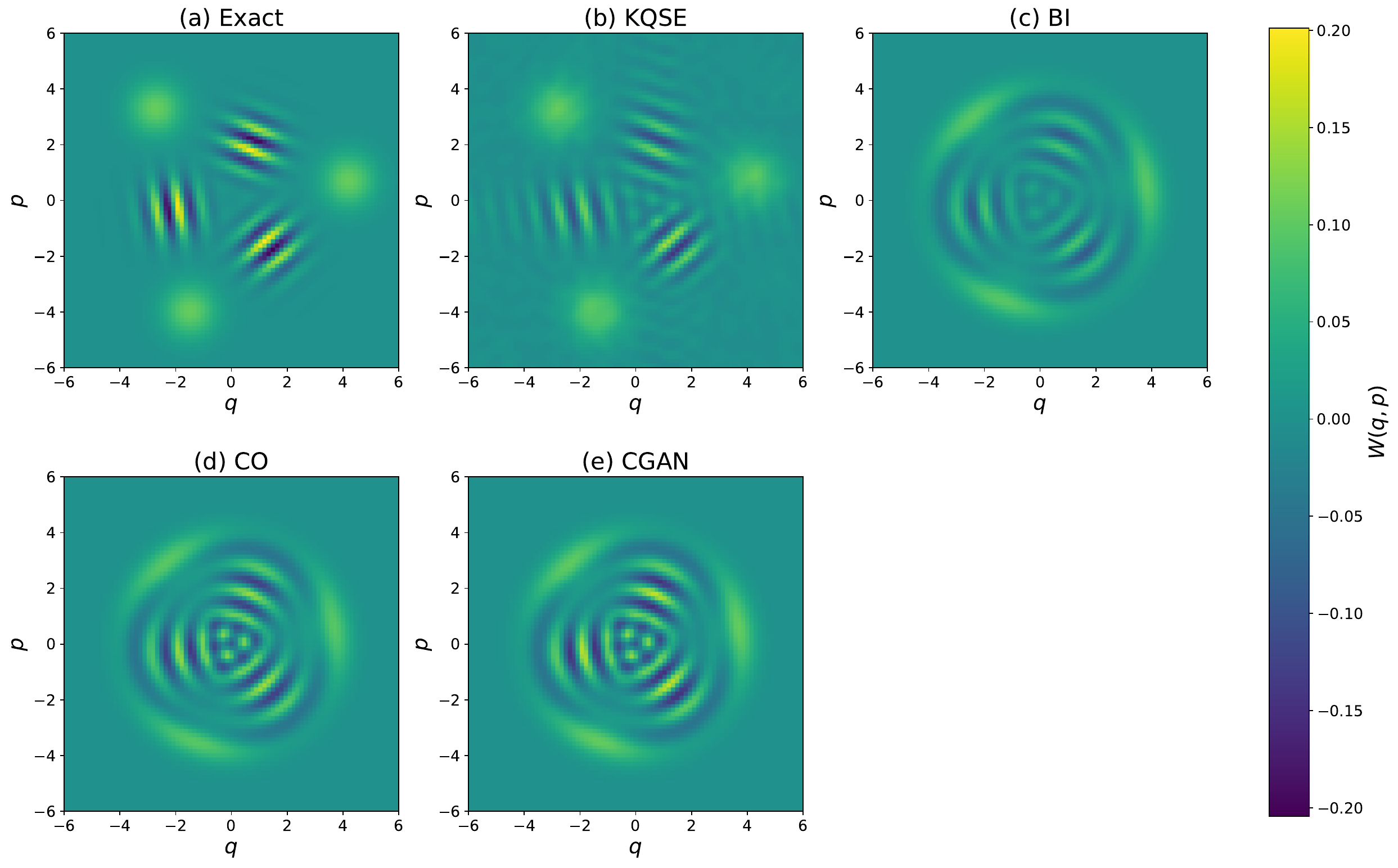}
\caption{
The Wigner function of the CCS with $a=3+0.5i$ compared to the BI, CO, and CGAN estimators with the Fock space cut-off
$N_{\mathrm{cut}}=10$, whereas KQSE is cut-off free.
}
\label{fig:hardcat_wigner_N10}
\end{figure}
The same data set is used for KQSE, Bayesian inference (BI), convex optimization (CO), and the conditional generative adversarial network (CGAN). For the finite dimensional reconstruction methods, we consider Fock-space cutoffs of $N_{\mathrm{cut}}=10$ and $N_{\mathrm{cut}}=20$. KQSE, by contrast, does not rely on a Fock-space truncation, and its reconstruction is therefore unchanged between the two comparisons. This difference is particularly important for the state considered here. With $|a|^2=9.25$, the photon-number distribution has non-negligible weight beyond the lowest Fock levels, making $N_{\mathrm{cut}}=10$ a rather restrictive approximation. The comparison at $N_{\mathrm{cut}}=20$ therefore tests the extent to which the finite dimensional methods benefit from access to a substantially enlarged reconstruction space.
\begin{figure}[ht]
\centering
\includegraphics[width=0.8\textwidth]
{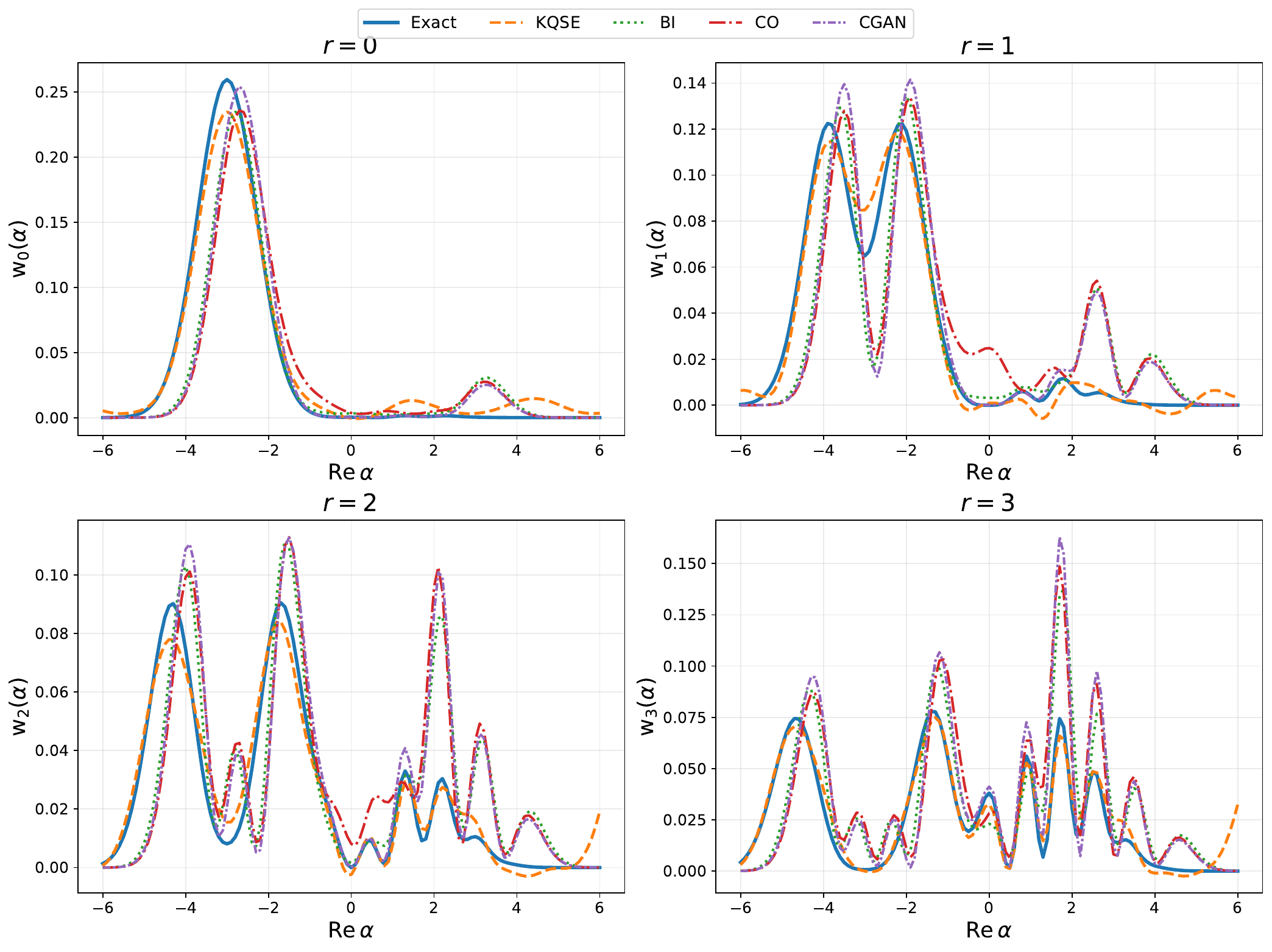}
\caption{
The Photon-number tomograms
$\mathsf w_r(\alpha)$, $r=0,1,2,3$, for the CCS with $a=3+0.5i$ compared to the BI, CO, and CGAN estimators with the Fock space cut-off
$N_{\mathrm{cut}}=10$, whereas KQSE is cut-off free.
The real displacement slice $\operatorname{Im}\alpha=0$ is shown.
}
\label{fig:hardcat_photon_N10}
\end{figure}
Since the benchmark is based on a single common data set rather than on an ensemble of independently generated realizations, we do not provide the statistical MISE. We  report the integrated squared error (ISE) for each individual reconstruction, as defined in Eq.~\eqref{eq:ISE_general}, together with the maximum squared pointwise discrepancy $E_{\infty,\Gamma_\Theta}$ over the finite evaluation grid. The latter corresponds to Eq.~\eqref{eq:empirical_uniform_pointwise_MSE} evaluated for $R=1$.
In Figure~\ref{fig:hardcat_wigner_N10}  the exact Wigner function is compared with the  KQSE, BI, CO, and CGAN for $N_{\mathrm{cut}}=10$ estimators. The corresponding photon-number tomograms are shown in
Fig.~\ref{fig:hardcat_photon_N10}.
Similar comparison is done for the $N_{cut}=20$ that is more beneficial for the parametric methods, while nothing is changed for the KQSE. The results are shown in Fig.~\ref{fig:hardcat_wigner_N20}. As expected, with the bigger cut the parametric methods perform better then with the smaller one. 
\begin{figure}[h]
\centering
\includegraphics[width=0.8\textwidth]
{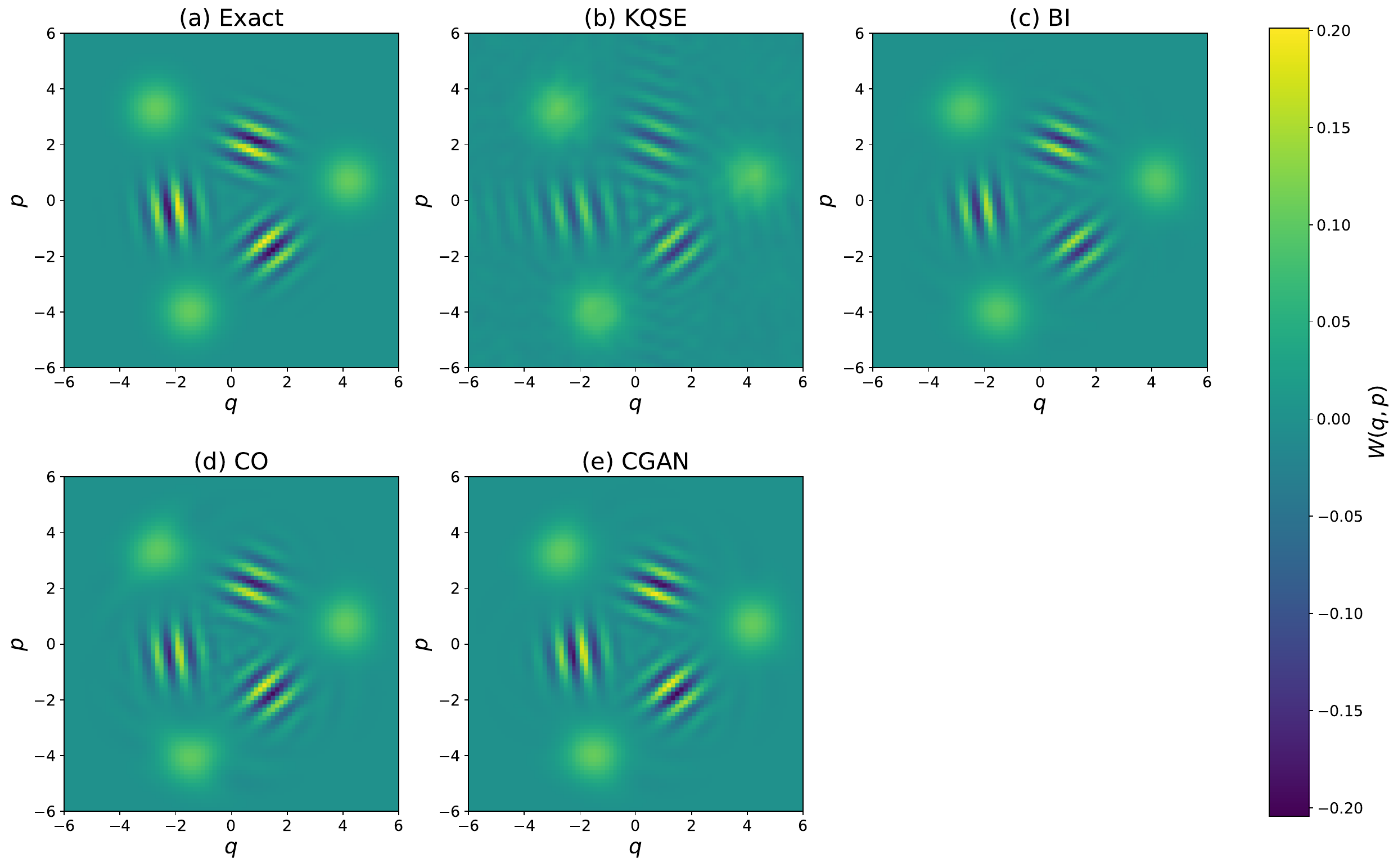}
\caption{
The Wigner function of the CCS with $a=3+0.5i$ compared to the BI, CO, and CGAN estimators with the Fock space cut-off
$N_{\mathrm{cut}}=20$, whereas 
KQSE is identical to that in Fig.~\ref{fig:hardcat_wigner_N10}, since
it does not require a Fock space cut-off.
}
\label{fig:hardcat_wigner_N20}
\end{figure}
\begin{figure}[h]
\centering
\includegraphics[width=0.8\textwidth]
{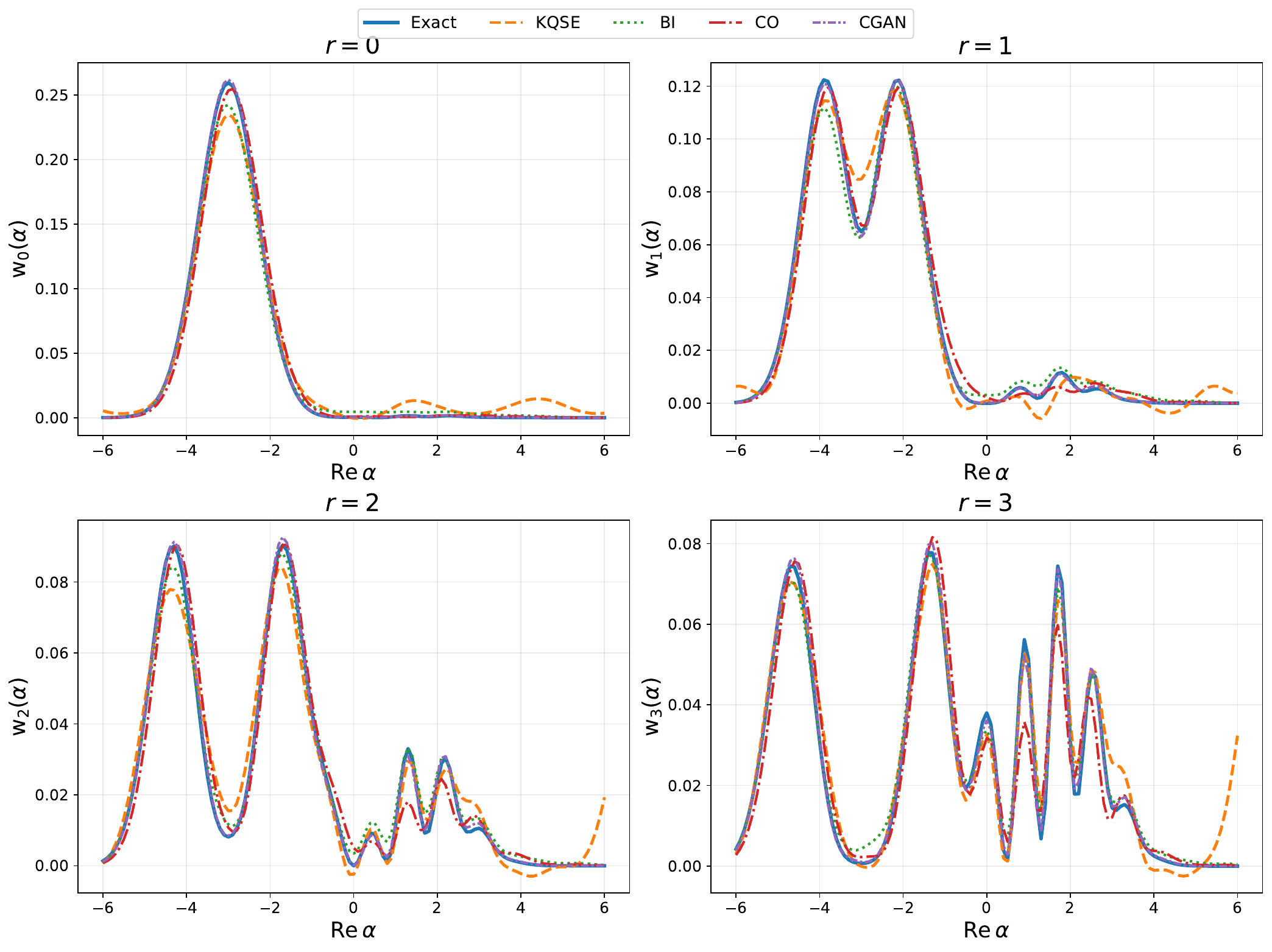}
\caption{
The photon-number tomograms
$\mathsf w_r(\alpha)$, $r=0,1,2,3$  of the CCS with $a=3+0.5i$ compared to the BI, CO, and CGAN estimators with the Fock space cut-off
$N_{\mathrm{cut}}=20$, whereas 
KQSE is identical to that in Fig.~\ref{fig:hardcat_photon_N10}, since
it does not require a Fock space cut-off.
}
\label{fig:hardcat_photon_N20}
\end{figure}
The numerical discrepancies are summarized in
Table~\ref{tab:hardcat_state_methods}. Each table entry contains
$\operatorname{ISE}/E_{\infty,\Gamma}$ for the corresponding
function and method. Bold entries denote the smallest value of both
metrics within the corresponding cut-off block.
\begin{table}[H]
\caption{
Estimation errors for the CCS with $a=3+0.5i$. Each numerical entry is reported as
$\operatorname{ISE}/E_{\infty,\Gamma}$.
The Fock space cut-off applies only to BI, CO, and CGAN while KQSE is
cut-off-free and consequently has the same values in both blocks.
Bold entries identify the smallest discrepancies for the corresponding
observable and cut-off.
}
\label{tab:hardcat_state_methods}
\centering
\small
\setlength{\tabcolsep}{5pt}

\resizebox{\textwidth}{!}{%
\begin{tabular}{@{}llcccc@{}}
\toprule
$\boldsymbol{N_{\mathrm{cut}}}$ &
\textbf{Function} &
\textbf{KQSE} &
\textbf{BI} &
\textbf{CO} &
\textbf{CGAN} \\
\midrule

\multirow{5}{*}{$10$}
& $W(q,p)$
& $\mathbf{2.323\times10^{-2}/1.022\times10^{-2}}$
& $1.109\times10^{-1}/3.312\times10^{-2}$
& $9.618\times10^{-2}/1.689\times10^{-2}$
& $1.012\times10^{-1}/1.830\times10^{-2}$ \\

& $\mathsf w_0(\alpha)$
& $\mathbf{1.073\times10^{-3}/6.15\times10^{-4}}$
& $8.780\times10^{-3}/8.963\times10^{-3}$
& $1.377\times10^{-2}/1.307\times10^{-2}$
& $1.123\times10^{-2}/1.127\times10^{-2}$ \\

& $\mathsf w_1(\alpha)$
& $\mathbf{4.378\times10^{-4}/4.10\times10^{-4}}$
& $5.468\times10^{-3}/3.549\times10^{-3}$
& $7.747\times10^{-3}/4.197\times10^{-3}$
& $7.455\times10^{-3}/5.147\times10^{-3}$ \\

& $\mathsf w_2(\alpha)$
& $\mathbf{3.859\times10^{-4}/3.73\times10^{-4}}$
& $5.153\times10^{-3}/3.160\times10^{-3}$
& $7.521\times10^{-3}/5.437\times10^{-3}$
& $7.124\times10^{-3}/5.176\times10^{-3}$ \\

& $\mathsf w_3(\alpha)$
& $\mathbf{4.439\times10^{-4}/1.048\times10^{-3}}$
& $5.325\times10^{-3}/4.591\times10^{-3}$
& $7.526\times10^{-3}/5.513\times10^{-3}$
& $7.696\times10^{-3}/7.827\times10^{-3}$ \\

\midrule

\multirow{5}{*}{$20$}
& $W(q,p)$
& $2.323\times10^{-2}/1.022\times10^{-2}$
& $5.454\times10^{-3}/2.300\times10^{-3}$
& $5.492\times10^{-3}/2.918\times10^{-3}$
& $\mathbf{1.009\times10^{-3}/4.52\times10^{-4}}$ \\

& $\mathsf w_0(\alpha)$
& $1.073\times10^{-3}/6.15\times10^{-4}$
& $4.37\times10^{-4}/3.58\times10^{-4}$
& $6.01\times10^{-4}/3.75\times10^{-4}$
& $\mathbf{6.0\times10^{-6}/6.0\times10^{-6}}$ \\

& $\mathsf w_1(\alpha)$
& $4.378\times10^{-4}/4.10\times10^{-4}$
& $1.84\times10^{-4}/1.44\times10^{-4}$
& $2.37\times10^{-4}/1.33\times10^{-4}$
& $\mathbf{6.0\times10^{-6}/3.0\times10^{-6}}$ \\

& $\mathsf w_2(\alpha)$
& $3.859\times10^{-4}/3.73\times10^{-4}$
& $8.5\times10^{-5}/3.6\times10^{-5}$
& $2.90\times10^{-4}/2.24\times10^{-4}$
& $\mathbf{1.6\times10^{-5}/1.3\times10^{-5}}$ \\

& $\mathsf w_3(\alpha)$
& $4.439\times10^{-4}/1.048\times10^{-3}$
& $8.1\times10^{-5}/3.2\times10^{-5}$
& $4.69\times10^{-4}/4.17\times10^{-4}$
& $\mathbf{2.6\times10^{-5}/1.9\times10^{-5}}$ \\

\bottomrule
\end{tabular}%
}
\end{table}
At the lower cut-off, $N_{\mathrm{cut}}=10$, KQSE gives the smallest
ISE and the smallest maximum squared pointwise discrepancy for every
observable considered.
This behaviour is consistent with the restrictive Fock space
truncation imposed on a CCS whose photon-number support
extends significantly beyond the lowest ten sectors.
The two cut-off regimes therefore illustrate an important distinction
between the reconstruction strategies. KQSE avoids a Fock space
truncation altogether: the same nonparametric KCFE
can be used for all target quantities and is unchanged when
the  cut-off of the competing methods is varied. This is
advantageous when a low dimensional truncation does not adequately
contain the state. On the other hand, when a sufficiently expressive
finite dimensional state space is available, the reconstructed
density matrix methods can exploit this additional representation
capacity and achieve substantially smaller finite data errors, as
demonstrated here by CGAN at $N_{\mathrm{cut}}=20$.

\subsection{Kitten state}
\label{subsec:experimental_study}
We next apply the same reconstruction pipeline to experimental homodyne
data for a non-Gaussian optical kitten state generated by conditional measurements at a highly unbalanced beam splitter from
Refs.~\cite{lvovsky2002quantumoptical,lvovsky2004iterative}. A weak coherent state of amplitude $\alpha$ and a conditionally prepared single-photon Fock state were injected into the two input ports, with reflectivity $r^2=0.925$ ($t=\sqrt{1-r^2}$). A single-photon detector in one output arm provided the conditioning: homodyne data in the other arm were acquired only upon a detector click. The resulting quantum interference removes which path information and prepares a nonclassical state in the signal mode. In the ideal limit, the conditional state takes the form
\begin{eqnarray}
\ket{\psi_s}=c_0\ket{0}+c_1\ket{1},\qquad
c_0=\frac{t}{\sqrt{t^2+\alpha^2}},\quad
c_1=\frac{\alpha}{\sqrt{t^2+\alpha^2}},
\end{eqnarray}
which is commonly referred to as a \textit{Schrödinger kitten state}.
The generated states were analysed by balanced homodyne detection with a phase scanned local oscillator, yielding quadrature probability distributions at multiple phases. The experimental data set contains
$T=14153$ homodyne samples
$\{x_{\theta_j},\theta_j\}_{j=1}^{T}$, where $x_{\theta_j}$ is the measured quadrature corresponding to the local-oscillator phase $\theta_j$.
The measurements are folded to $\theta\in[0,\pi)$ using
$x_{\theta+\pi}=-x_\theta$ and grouped into $N_\theta=20$ nearly equally
populated phase bins.
\par To account phenomenologically for optical losses, finite single-photon preparation efficiency and detector dark counts, the reference state is modeled as
\begin{equation}
\label{eq:mixed_kitten_model}
\boldsymbol{\rho}
=
(1-p)|0\rangle\langle0|
+
p|\psi_s\rangle\langle\psi_s|, 
\end{equation}
where for $\alpha=0.3$ one can estimate $p$ from the reported purity to be equal to $p=0.6810463$. The corresponding symplectic tomogram is
\begin{equation}
\label{eq:kitten_tomogram_model}
\mathcal W(x|\mu,\nu)
=
\mathcal W_0(x|\beta)
\left[
1-pc_1^2
+
2pc_1^2\frac{x^2}{\beta^2}
+
2\sqrt{2}\,
pc_0c_1
\frac{\mu x}{\beta^2}
\right],
\qquad
\beta=\sqrt{\mu^2+\nu^2},
\end{equation}
where the tomogram  of the ground state of the harmonic oscillator is given by the formula \eqref{1530}. The CF is
\begin{equation}
\label{eq:kitten_CF_model}
\phi(t;\mu,\nu)
=
\exp\!\left[
-\frac{t^2(\mu^2+\nu^2)}{4}
\right]
\left[
1
-
\frac{
pc_1^2t^2(\mu^2+\nu^2)
}{2}
+
i\sqrt{2}\,
pc_0c_1\mu t
\right].
\end{equation}
The corresponding Wigner function is
\begin{equation}
\label{eq:kitten_Wigner_model}
W(q,p)
=
\frac{e^{-(q^2+p^2)}}{\pi}
\left[
1
-
2pc_1^2
+
2pc_1^2(q^2+p^2)
+
2\sqrt{2}\,
pc_0c_1q
\right].
\end{equation}
The photon-number tomogram of the state is
\begin{align}
\label{eq:kitten_photon_tomogram_model}
\mathsf w_r(\alpha)
&=
\mathsf w_{0r}(\alpha)
\Bigg[
1-p|c_1|^2
+
p|c_1|^2
\frac{(r-|\alpha|^2)^2}{|\alpha|^2}+2p
\frac{r-|\alpha|^2}{|\alpha|^2}
\operatorname{Re}
\left(
c_0c_1^*\alpha
\right)
\Bigg],
\end{align}
where
\begin{equation}
\label{eq:vacuum_photon_tomogram}
\mathsf w_{0r}(\alpha)
=
e^{-|\alpha|^2}
\frac{|\alpha|^{2r}}{r!}, \qquad r=0,1,2,\dots
\end{equation}

We consider two variants of KQSE. The uncorrected reconstruction and the one with the noise correction. The noise correction procedure is described in detail in the Supplemental Material of Ref.~\cite{markovich2025nonparametric}.
Figure~\ref{fig:exp_wigner_all_methods} shows the estimators of the Wigner function with the two KQSE variants, BI, CO and CGAN.
\begin{figure}[h]
\centering
\includegraphics[width=0.8\textwidth]
{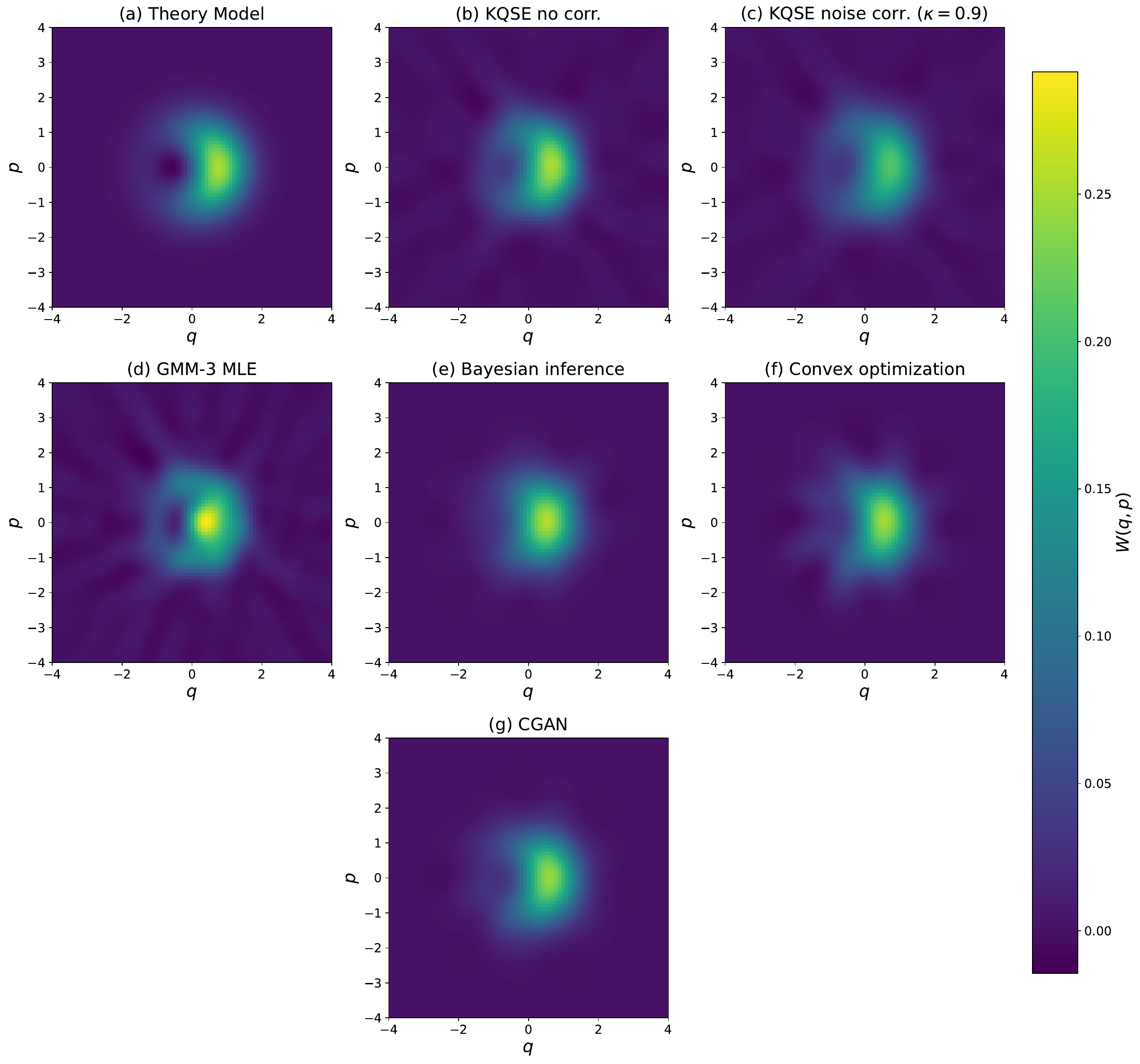}
\caption{The estimates of the Wigner function of the kitten state.
Top row:
(a) theoretical Wigner function,
(b) KQSE without noise correction, and
(c) noise-corrected KQSE.
Middle row:
(d) GMM-3 MLE,
(e) BI,
(f) CO.
Bottom centered panel:
(g) CGAN.
The parameter $N_{cut}=10$ for BI, CO, CGAN.
}
\label{fig:exp_wigner_all_methods}
\end{figure}
The same set of reconstructions is compared for the
photon-number tomograms in
Fig.~\ref{fig:exp_photon_all_methods}.
\begin{figure}[ht]
\centering
\includegraphics[width=0.8\textwidth]
{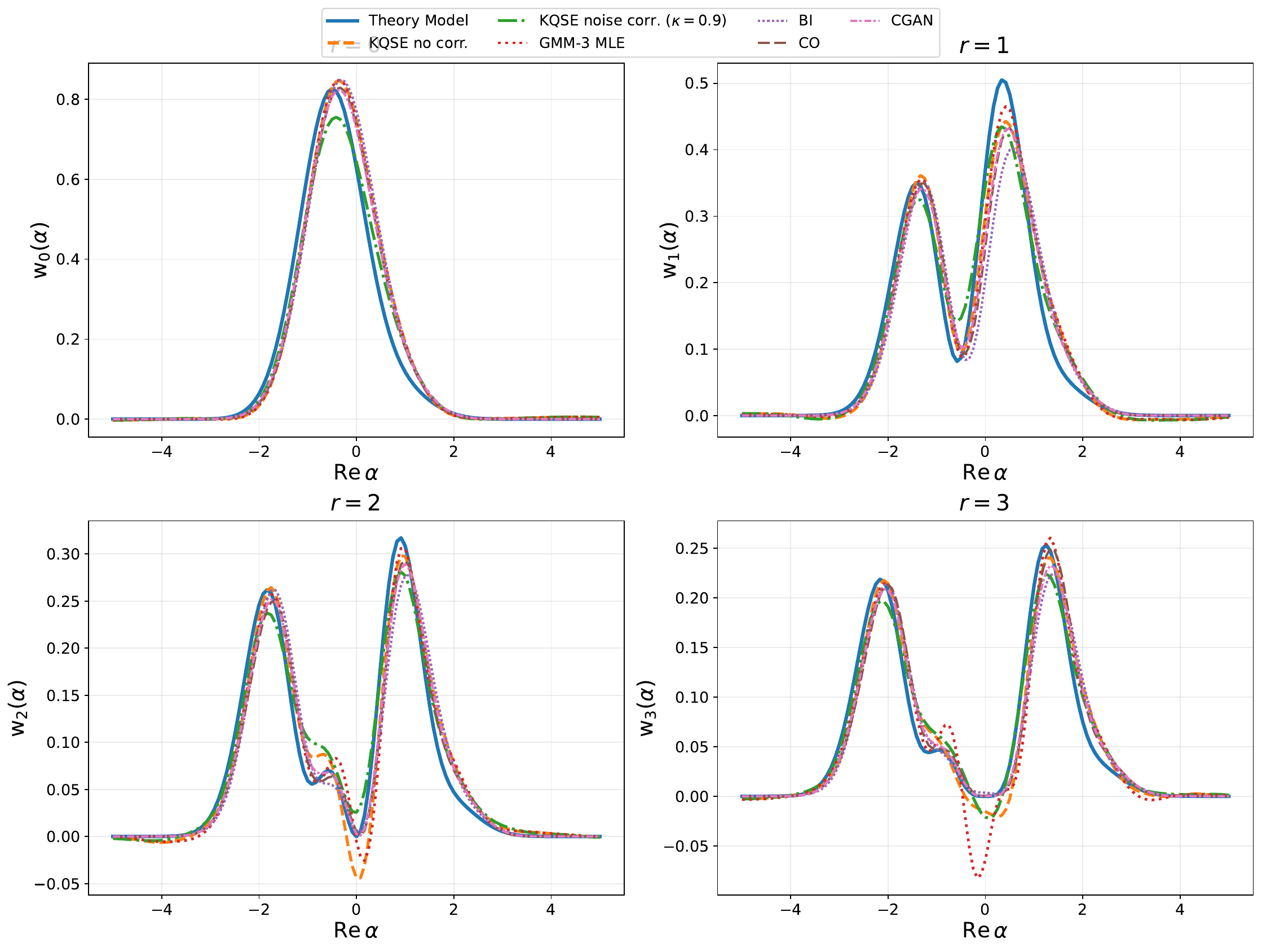}
\caption{
The photon-number tomograms
$\mathsf w_r(\alpha)$, $r=0,1,2,3$, compared to the BI, CO, and CGAN estimators with the Fock space cut-off
$N_{\mathrm{cut}}=10$, whereas KQSE is cut-off free.
}
\label{fig:exp_photon_all_methods}
\end{figure}

For the Wigner function, the experimental model ISE is evaluated as
\begin{equation}
\label{eq:experimental_ISE_W}
\widehat{\operatorname{ISE}}_{W}
=
\Delta q\,\Delta p
\sum_{i,j}
\left|
\widehat W(q_i,p_j)
-
W(q_i,p_j)
\right|^2,
\end{equation}
whereas for the photon-number tomograms
\begin{equation}
\label{eq:experimental_ISE_wr}
\widehat{\operatorname{ISE}}_{\mathsf w_r}
=
\Delta\alpha
\sum_i
\left|
\widehat{\mathsf w}_r(\alpha_i)
-
\mathsf w_r(\alpha_i)
\right|^2.
\end{equation}
The resulting ISE values and squared uniform model discrepancies
 for the two KQSE variants and the four
competing reconstruction methods are collected in
Table~\ref{tab:exp_all_method_comparison}.
\begin{table}[h]
\caption{
Experimental discrepancies relative to the approximate kitten state
reference model for KQSE without noise correction, noise corrected
KQSE, GMM-3 MLE, BI, CO, and CGAN.
BI, CO, and CGAN use the Fock space
cut-off $N_{\mathrm{cut}}=10$ of the original experimental benchmark.
Bold entries denote the smallest value for each quantum-state representations and metric.
}
\label{tab:exp_all_method_comparison}
\centering
\footnotesize
\setlength{\tabcolsep}{3pt}

\begin{tabularx}{\textwidth}{@{}CCCC@{}}
\toprule
\textbf{Function} &
\textbf{Method} &
$\boldsymbol{\widehat{\operatorname{ISE}}}$ &
$\boldsymbol{\widehat E_{\infty,\Gamma}}$ \\
\midrule

\multirow[m]{6}{*}{$W(q,p)$}
& KQSE, no corr. &
$6.862\times10^{-3}$ &
$3.832\times10^{-3}$ \\
& KQSE, noise corr. ($\kappa=0.9$) &
$\mathbf{4.510\times10^{-3}}$ &
$\mathbf{2.546\times10^{-3}}$ \\
& GMM-3 MLE &
$9.338\times10^{-3}$ &
$1.336\times10^{-2}$ \\
& BI &
$1.169\times10^{-2}$ &
$1.109\times10^{-2}$ \\
& CO &
$7.437\times10^{-3}$ &
$5.492\times10^{-3}$ \\
& CGAN &
$6.035\times10^{-3}$ &
$4.608\times10^{-3}$ \\
\midrule

\multirow[m]{6}{*}{$\mathsf w_0(\alpha)$}
& KQSE, no corr. &
$2.655\times10^{-2}$ &
$2.217\times10^{-2}$ \\
& KQSE, noise corr. ($\kappa=0.9$) &
$\mathbf{1.362\times10^{-2}}$ &
$\mathbf{8.469\times10^{-3}}$ \\
& GMM-3 MLE &
$2.468\times10^{-2}$ &
$1.990\times10^{-2}$ \\
& BI &
$3.430\times10^{-2}$ &
$2.986\times10^{-2}$ \\
& CO &
$2.274\times10^{-2}$ &
$1.750\times10^{-2}$ \\
& CGAN &
$2.044\times10^{-2}$ &
$1.638\times10^{-2}$ \\
\midrule

\multirow[m]{6}{*}{$\mathsf w_1(\alpha)$}
& KQSE, no corr. &
$8.713\times10^{-3}$ &
$5.677\times10^{-3}$ \\
& KQSE, noise corr. ($\kappa=0.9$) &
$\mathbf{6.574\times10^{-3}}$ &
$\mathbf{4.953\times10^{-3}}$ \\
& GMM-3 MLE &
$8.860\times10^{-3}$ &
$6.084\times10^{-3}$ \\
& BI &
$2.534\times10^{-2}$ &
$3.029\times10^{-2}$ \\
& CO &
$1.364\times10^{-2}$ &
$1.425\times10^{-2}$ \\
& CGAN &
$1.123\times10^{-2}$ &
$1.148\times10^{-2}$ \\
\midrule

\multirow[m]{6}{*}{$\mathsf w_2(\alpha)$}
& KQSE, no corr. &
$4.590\times10^{-3}$ &
$3.093\times10^{-3}$ \\
& KQSE, noise corr. ($\kappa=0.9$) &
$\mathbf{2.937\times10^{-3}}$ &
$\mathbf{2.056\times10^{-3}}$ \\
& GMM-3 MLE &
$5.153\times10^{-3}$ &
$6.271\times10^{-3}$ \\
& BI &
$6.990\times10^{-3}$ &
$6.330\times10^{-3}$ \\
& CO &
$4.114\times10^{-3}$ &
$3.547\times10^{-3}$ \\
& CGAN &
$4.051\times10^{-3}$ &
$3.367\times10^{-3}$ \\
\midrule

\multirow[m]{6}{*}{$\mathsf w_3(\alpha)$}
& KQSE, no corr. &
$3.066\times10^{-3}$ &
$1.674\times10^{-3}$ \\
& KQSE, noise corr. ($\kappa=0.9$) &
$\mathbf{1.869\times10^{-3}}$ &
$\mathbf{1.037\times10^{-3}}$ \\
& GMM-3 MLE &
$6.046\times10^{-3}$ &
$6.842\times10^{-3}$ \\
& BI &
$4.456\times10^{-3}$ &
$3.759\times10^{-3}$ \\
& CO &
$3.390\times10^{-3}$ &
$2.345\times10^{-3}$ \\
& CGAN &
$2.927\times10^{-3}$ &
$2.313\times10^{-3}$ \\
\bottomrule
\end{tabularx}
\end{table}
This result should be interpreted as evidence that the fixed
noise correction model improves agreement with the experimental
kitten-state benchmark.

\section{Discussion and Conclusions}
\label{sec:conclusions}

In this work, we extended Kernel Quantum State Estimation to a general class of kernel integral transforms of the tomographic characteristic function, covering all phase-space quasiprobability distributions, e.g., the Wigner and Husimi functions, as well as photon-number tomograms. We prove that the total reconstruction error of the KQSE for this integral transformation achieves nearly optimal uniform pointwise MSE scaling $\widetilde O(T^{-1})$, where $T$ is the total amount of homodyne/heterodyne measurements. 

The numerical results support this results. In the simulation study, KQSE avoided the model bias saturation observed for the MLE baseline, while the experimental study demonstrated that the same estimated characteristic function can be used to reconstruct several phase-space functionals directly from homodyne data. These results suggest that the principal advantage of KQSE on experimental data is not necessarily a universal reduction of every finite sample error, but rather its model independent applicability to several state representations within one reconstruction procedure. 

Further developments should include uncertainty quantification, detector inefficiency and other measurement imperfections, and physicality preserving corrections that retain the nonparametric character of the method. Extensions to multimode continuous variable systems and adaptive allocation of measurements among tomographic directions are particularly important, although they will require controlling the growth of both the integration dimension and the measurement budget. More broadly, the kernel transform viewpoint developed here suggests that a single nonparametric estimate of the tomographic characteristic function can serve as a reusable statistical representation of the quantum state from which a wide family of experimentally relevant quantities can be evaluated with explicit error control.

\authorcontributions{Conceptualization, V.A. and L.A.; methodology, V.A. and L.A.; software, V.A.; validation, V.A. and L.A.; formal analysis, V.A. and L.A.; investigation, V.A. and L.A.; resources,V.A. and L.A.; data curation, V.A. and L.A.; writing---original draft preparation, V.A. and L.A.; writing---review and editing, V.A. and L.A.; visualization, V.A.; supervision, L.A.; project administration, L.A.; funding acquisition, L.A. All authors have read and agreed to the published version of the manuscript.}

\funding{L.M. was  supported by the Netherlands Organisation for Scientific Research (NWO/OCW), as part of the Quantum Software Consortium program (project number 024.003.037 / 3368).}

\institutionalreview{Not applicable.}
\informedconsent{Not applicable.}

\dataavailability{We encourage all authors of articles published in MDPI journals to share their research data. In this section, please provide details regarding where data supporting reported results can be found, including links to publicly archived datasets analyzed or generated during the study. Where no new data were created, or where data is unavailable due to privacy or ethical restrictions, a statement is still required. Suggested Data Availability Statements are available in section ``MDPI Research Data Policies'' at \url{https://www.mdpi.com/ethics}.}

\acknowledgments{We thank A. Lvovsky for providing the experimental homodyne-tomographic data used in this work and for valuable discussions.}

\conflictsofinterest{The authors declare no conflicts of interest.} 


\reftitle{References}

\externalbibliography{yes}
\bibliography{tomogram}

\PublishersNote{}
\isPreprints{}{
\end{adjustwidth}
} 
\end{document}